\documentclass{article}

\usepackage[utf8]{inputenc} % allow utf-8 input
\usepackage[T1]{fontenc}    % use 8-bit T1 fonts
\usepackage{hyperref}       % hyperlinks
\usepackage{url}            % simple URL typesetting
\usepackage{booktabs}       % professional-quality tables
\usepackage{amsfonts}       % blackboard math symbols
\usepackage{nicefrac}       % compact symbols for 1/2, etc.
\usepackage{microtype}      % microtypography
\usepackage{xcolor}         % colors

\usepackage{algorithm}
\usepackage[noend]{algorithmic}
\usepackage{float}

\usepackage[a4paper,outer=2.0cm,inner=2.0cm,top=2cm,bottom=2cm]{geometry} % per definizione layout
\usepackage{mathtools}
\usepackage{amssymb}
\usepackage{amsthm}
\usepackage{amsmath}
\usepackage{mathtools}
\usepackage{thmtools}
\usepackage{thm-restate}

\usepackage{multicol}
\usepackage{caption}
\usepackage{subcaption}
\usepackage{chngcntr}
\usepackage{bbm}
\usepackage{hyperref}
\usepackage{verbatim}
\usepackage[toc,page]{appendix}
\usepackage{xcolor}
\usepackage[capitalize,noabbrev]{cleveref}

\newcommand{\wt}{\widetilde}

\newtheorem{theorem}{Theorem}[section]
\newtheorem{definition}[theorem]{Definition}
\newtheorem{lemma}[theorem]{Lemma}
\newtheorem*{lemma*}{Lemma}
\newtheorem*{definition*}{Definition}

\newtheorem{fact}[theorem]{Fact}

\DeclareMathOperator*{\dist}{dist}

\DeclareMathOperator*{\E}{\mathbb{E}}

\newcommand{\lsh}{\textsc{QueryLSH}}

\newcommand{\coll}{\textsc{Coll}}

\title{On the adversarial robustness of Locality-Sensitive Hashing in Hamming space}
\date{}

\author{%
  Michael Kapralov \\
  \and
  Mikhail Makarov \\
  \and
  Christian Sohler \\
}

\begin{document}
% \DontPrintSemicolon

\maketitle

\begin{abstract}
    Locality-sensitive hashing~\cite{IndykM98} is a classical data structure for approximate nearest neighbor search. It allows, after a close to linear time preprocessing of the input dataset, to find an approximately nearest neighbor of any fixed query in sublinear time in the dataset size.  The resulting data structure is randomized and succeeds with high probability for every fixed query. 
    In many modern applications of nearest neighbor search the queries are chosen adaptively. In this paper, we study the robustness of the locality-sensitive hashing to adaptive queries in Hamming space. We present a simple adversary that can, under mild assumptions on the initial point set, provably find a query to the approximate near neighbor search data structure that the data structure fails on. Crucially, our adaptive algorithm finds the hard query exponentially faster than random sampling.
\end{abstract}

\section{Introduction}

Nearest neighbor search is one of the most fundamental data structure problems in machine learning. In particular, the development of efficient data structures for nearest neighbor search in high dimensions is a challenge
%
%
%A major challenge towards the development of efficient data structures for (approximate) nearest neighbor search 
%for high-dimensional point sets is often called the \emph{curse of dimensionality}.
%Essentially, the higher the dimension of the space, the more points can have similar distances to each other. This makes the development of data structures for high dimensional nearest neighbor search a particular challenge and 
and finding the exact nearest neighbor often requires a linear scan over the data~\cite{WeberSB98}. Therefore, people also started to investigate approximate nearest neighbor search.
 One important technique to design approximate nearest neighbor data structures is \emph{locality sensitive hashing (LSH)} which has been introduced in the seminal paper by Indyk and Motwani \cite{IndykM98}. LSH is widely used in many applications including
music retrieval \cite{RK08}, video anomaly detection \cite{ZLZRS16}, and clustering \cite{KIW07} (see also the survey \cite{JMNIC21} for further examples). 
%LSH does not directly solve the nearest neighbor problem, but the related near neighbor problem, which can be stated as follows: For parameters $r$ and $c\geq 1$, the objective is to report a point within distance $cr$ to a query $q$, if a point within distance $r$ to $q$ exists. Once this problem can be solved, there is a reduction that implies that one can also solve the nearest neighbor problem.

%t
%requires a class of hash functions that satisfies the $(r_1,r_2,p_1,p_2)$-sensitivity property \cite{}, i.e. collisions between near points (at distance at most $r_1$) happen with probability at least $p_1$ while collisions between points whose distance is larger than $r_2$ happen with probability at most $p_2<p_1$.
%It then combines such hash functions to amplify the gap between the two collision probabilities. Finally, it uses a sublinear number of hash functions to guarantee that with good probability a query collides with an approximate near neighbor under some hash function, if a near neighbor exists. 
%Several instances of the near neighbor problem can be combined to solve the approximate nearest neighbor problem.

LSH is a randomized data structure and as such its performance guarantees depend on its internal randomness. The standard analysis assumes an \emph{oblivious} adversary, i.e. queries may not depend on answers to earlier queries for the guarantees to apply.
%Typically, LSH guarantees that for a fixed query $q$, the data structure returns a near neighbor with probability at least $1-\delta$ for some failure probability $\delta$. Using a standard union bound these guarantees carry over to sets of queries, \textbf{if} these queries are non-adaptive, i.e. the decision which query to make is not allowed to be a function of the answers of earlier queries.
%
What are the consequences, if we would like to have the LSH performance guarantees when LSH is used as a data structure in an algorithm? It essentially says that for the guarantees to hold, the algorithm must be able to specify at a single point of time a set $Q$ of query points and it gets all answers to these query points at once. Then it can do any computation with the answers but it may not ask further queries. 
 Such a restricted functionality significantly limits the strength of the guarantees provided by LSH as more often than not future queries depend on the result of earlier queries. 
 Therefore, it would be great to better understand the LSH guarantees in the adaptive setting.
 Even though some work on LSH such as constructions without false negatives \cite{Pagh16,Ahle17} (see also Related Work for further discussions) 
 are motivated by this question, there is 
 still a lack of knowledge about how resilient standard LSH data structures
 are if we allow adaptive adversaries.
 
 In this paper, we address this question by developing a randomized query scheme (adversarial algorithm) for LSH under the Hamming distance with the goal of efficiently generating a query that makes LSH return an incorrect answer. We show that for a wide setting of parameters, if a point set has a somewhat isolated point, then we can make LSH have a false positive (that is, returning no point while there is a near point) much faster than this would happen with non-adaptive queries. 

Let us discuss this further on a simple example. We may use LSH for an approximate version of a $1$-nearest neighbor classifier (see \cite{ABVT20} for a related study), that is, we store the training set in an LSH data structure and on a query we return the label of the (approximate) nearest point in the training set that is returned by LSH.
%Recall that a $1$-nearest neighbor classifier gets a labeled test set and for every query returns the label of the nearest instance of the test set. 
%If the data set is large we may think of using LSH as a way to get an approximate nearest neighbor classifier (the $1$-NN classifier may be a bit simplistic, but there are several examples where LSH has been used to get an approximate $k$-NN classifier. 
In this setting LSH will satisfy its standard guarantees when the choice of queries does not depend on previous answers. This happens, for example, when queries are independently sampled from a fixed query distribution.
However, one could potentially design a carefully engineered sequence of queries to construct queries on which the algorithm returns an incorrect answer. This is particularly problematic if the algorithm is used for a sensitive task and relies on strong guarantees.

\textbf{Our results.} Our main result is an algorithm for generating an adversarial query:

%\begin{theorem}[Informal version of Theorem~\ref{thm:simple}] \label{thm:simple-inf}
% Given an `isolated point' $z$ in a dataset $P$,  one can find a point $q$ at a distance at most $r$ from $z$ such that querying LSH with $q$ returns no point. The number of queries to the LSH data structure needed to generate the point $q$ is bounded by $O( cr \cdot \log(1/\delta))$, where $\delta$ is the failure probability of LSH.
%\end{theorem}
%
%Note that crucially, the dependence on the inverse failure probability is logarithmic. We can also reduce the dependence on the far radius $cr$ to logarithmic:

\begin{theorem}[Informal version of Theorem~\ref{thm:fast}] \label{thm:fast-inf}

 Given an `isolated point' $z$ in a dataset $P$, one can find a point $q$ at a distance at most $r$ from $z$ such that querying LSH with $q$ returns no point. The number of queries to the LSH data structure needed to generate the point $q$ is bounded by $O(\log (cr) \cdot \log(1/\delta))$, where $\delta$ is the failure probability of LSH.
\end{theorem}

Note that to find such a point through random sampling, one would have to make $\Omega(1/\delta)$ queries in expectation \cref{thm:lsh-guarantee}, which is exponentialy slower in $\delta$ than the proposed adversary. We then evaluate our data structure experimentally in various settings to understand how the efficiency and effectiveness of our approach depends on the various parameters.

\section{Related Work} \label{sec:related}

\textbf{Monte-Carlo Locality-Sensitive Hashing.}
Most of the work dedicated to Locality-Sensitive Hashing was done in the Monte-Carlo model, which we consider to be the standard model in this paper. That is, the goal is to construct a data structure with deterministically bounded query time that correctly answers the queries with high probability.
Constructions of locality-sensitive hash families are known for several different metrics, including Hamming distance~\cite{IndykM98,AndoniR15}, $\ell_1$~\cite{datar2004locality}, the Euclidean metric~\cite{AndoniI06}, Jaccard similarity~\cite{Charikar02,DahlgaardKT17}, angular distance~\cite{AndoniILRS15,KennedyW17,Linips12} and more.
Asymmetric LSH solutions, which allow to reduce query time at the expense of an increase in space, or reduce space at the expense of increasing query time, have received significant attention in the literature and are quite well understood by now~\cite{Panigrahy06,Kapralov15,Christiani17,AndoniLRW17}. Distance sensitive hash families, where collision probabilities are not necessarily monotone in the distance between points being hashed, have also been constructed~\cite{0001CP018}.
%
 %The problem of proving lower bounds for the nearest-neighbor search has received a lot of attention in the literature. The results of~\cite{MotwaniNP06,ODonnellWZ14} give lower bounds for the $\rho$ parameter of locality sensitive hash families.  The results of~\cite{PanigrahyTW08,PanigrahyTW10} give lower bounds in the cell probe model. 
 %
 Data-dependent results, which adapt the locality-sensitive hash family to the dataset to achieve better exponents, are also known~\cite{AndoniINR14,AndoniLRW17}. 
%A natural question is whether standard LSH has false negatives. We show that this is indeed the case for some reasonable range of parameters.

\textbf{Las-Vegas Locality Sensitive Hashing.}
% As we discussed in the introduction, false negatives will result in unwanted effects. Although the standard LSH guarantees tell us that false negatives do not occur very often, this guarantee only holds for non-adaptive queries, i.e. queries that do not depend on the answer of earlier queries. One way to mitigate this problem would be to design
% an LSH scheme that does not have false negatives. 
In an attempt to avoid the problem of having false negatives,
another line of work seeks to construct LSH instances that answer the queries deterministically correctly, but the running time is a random variable and guarantees are usually on its expectation.
It was first shown by \cite{Pagh16} that such a scheme exists, however this scheme for general parameter settings requires more space than the standard LSH. A later work by \cite{Ahle17} achieves the same space-time trade-off as in \cite{IndykM98}. \cite{Wei22} further improves on it for $\ell_p$, and offers a data-dependent construction following \cite{AndoniI06}.
%
%\xxx[MM]{Should I say more?}
%
Although Las Vegas constructions avoid the problem of having false negatives, their guarantees still hold only against an oblivious adversary. It is entirely possible that an adaptive adversary can find a sequence of queries with an average response time much larger than what is guaranteed against an oblivious adversary. 
In any case, investigation of such adversaries lies outside of the scope of this paper.

%\xxx[MM]{I didn't find other papers...}

\textbf{Robustness to adaptive adversaries.}
One of the main motivations of this paper is to study the behavior of a classic data structure solving against an adaptive adversary. One of the most studied class of algorithms that is relevant to our paper is linear sketches. Some of the notable early works here include \cite{HardtW13}, which proposed an efficient attack on the linear sketch for norm estimation, and \cite{NelsonNW12}, which proposed deterministic sketch constructions for a range of problems that work for all inputs and provide lower bounds on their size. In \cite{KushilevitzOR00} an adversarially robust approach to solving approximate nearest neighbor in Hamming space is proposed, which does not use LSH. It is improved upon by \cite{CherapanamjeriN20}, which gives a robust construction for vector $\ell_p$-norm estimation, for $0 < p \leq 2$. On the other side of the spectrum, an adaptive attack on CountSketch was proposed by \cite{CohenNSS23,Cohen0NSSS22}, with the latter also proposing a robustifycation of this sketch.

Another fruitful line of work covers the use of differential privacy to make existing algorithms more robust against adaptive adversaries \cite{BeimelKMNSS22,WoodruffZ21,HassidimKMMS22,Ben-EliezerJWY22,AttiasCSS23}. The main idea is to use differentially private mechanisms to hide the random bits of the algorithm from the adversary. This approach turned out to be applicable to a wide range of estimation problems.

%\xxx[MM]{TODO: fill with discussion of previous works on adv. robustness of datastructures}

\section{Preliminaries}
\label{sec:prelims}
% \subsection{Notation}
For $x > 0$, $\log x := \log_2 x$ will denote binary logarithm of $x$. For $n \in \mathbb{N}$, $[n] := \{ 1, \ldots, 
n\}$ will be the set of integers from $1$ to $n$. For a point $p \in \{0, 1\}^d$ and $i \in [d]$, $p_i$ is the value of the point $p$ along $i$-th dimension. For two points $p, q \in \{0, 1\}^d$, $\dist(p, q) = \sum_{i \in [d]} |p_i - q_i|$ is their Hamming distance.

\textbf{Basics of Locality Sensitive Hashing.}
\label{sec:basics}
Locality sensitive hashing is a technique to derive randomized data structures for near neighbor search. This can be formalized as follows.
A (randomized) data structure for the $(c,cr)$-approximate
near neighbor problem in a metric space $(X,\dist)$ preprocesses
a point set $P\subset X$ in such a way that for a query point $q\in X$ with probability at least $1-\delta$ it (a) returns $p \in P$ with $\dist(p,q) \le cr$, if there is $p'\in P$
with $\dist(p',q) \le r$,
(b) returns $ \bot$, if there is no point $p\in P$ with $\dist(p,q) \le cr$, and
(c)
returns either $\bot$ or $p\in P$ with $\dist(p,q) \le cr$, otherwise.
Here $\delta$ is the failure probability of the data structure.
In this paper, we will consider the Hamming space, i.e. $X=\{0,1\}^d$ and $\dist$ is the Hamming distance and in the following discussion we focus on this setting.

%Locality Sensitive Hashing (LSH) is a technique to derive a randomized data structure for the approximate near neighbor problem. For the Hamming space this problem can be formalized as follows:
%\begin{definition}[$(r, cr)$-approximate near neighbor problem]
%$    In the $(r, cr)$-approximate near neighbor problem the input is a dataset $P \subseteq \{0, 1\}^d$ of $n$ points as well as a
%    query point $q \in \{0,1\}^d$. The objective is to compute a point 
%     $p \in P$ such that $\dist(p, q) \leq cr$ if there exists $p' \in P$ such that $\dist(q, p') \leq r$ and 
%\end{definition}

%A (randomized) data structure for the $(r,cr)$-near neighbor problem processes the point set $P$ such that for a given query point $q$ the data structure returns with probability at least 
%$1-\delta$ a 

The key component of LSH schemes is a locality-sensitive hash family. The idea behind them is that a function from such a family is more likely to map close points to the same bucket than far ones.

\begin{definition}[\cite{IndykM98}]
    A hash family $\mathcal{H}$ is $(r, cr, p_1, p_2)$-sensitive if for every $p, q \in \{0, 1\}^d$, the following holds. Let $h$ be chosen uniformly at random from $\mathcal{H}$. If $\dist(p, q) \leq r$, then $\Pr[h(p) = h(q)] \geq p_1$. Otherwise, if $\dist(p, q) \geq cr$, $\Pr[h(p) = h(q)] \leq p_2$.
\end{definition}

For a $d$-dimensional Hamming space $\{0, 1\}^d$, the locality-sensitive hash family has a very simple form \cite{IndykM98}. Let 
$ \label{eq:def_H}
    \mathcal{H} = \{h^{(i)}, i \in [d]: \forall p \in \{0, 1\}^d, h^{(i)}(p) = p_i\}
$
be the elementary hash function family. That is, it consists of hash functions $h^{(i)}:\{0, 1\}^d \to \{0, 1\}$, which project points from $\{0, 1\}$ to their $i$-th bit. It is known that

\begin{fact}[\cite{IndykM98}]
For any $c$, $r$ such that $r < cr \leq d$, the elementary hash function family $\mathcal{H}$ is $(r, cr, 1 - \frac{r}{d}, 1 - \frac{cr}{d})$-sensitive.
\end{fact}

Indyk and Motwani give the following guarantee (assuming that the points come from a $d$-dimensional space):
\begin{theorem}[\cite{IndykM98}] \label{thm:lsh-guarantee}
Suppose there exists a $(r,cr,p_1,p_2)$-sensitive family of hash functions. Then there is a data structure for the $(r,cr)$-approximate near neighbor problem (with $\delta=1/2$) that uses $O(dn+n^{1+\rho})$ space and requires $O(n^\rho)$
evaluations of hash functions, where $\rho = \frac{\log (1 - r/d)}{\log (1 - cr/d)}$.
\end{theorem}

In the following, we briefly describe the construction of an LSH data structure, which can be used to prove the above theorem. An LSH data structure consists of a multiset of $G$ of $L$ hash functions in the  product $\mathcal{H}^k=(h_1, h_2, \ldots, h_k), h_j\in \mathcal{H},j\in \{1, 2,\ldots, k\},$ i.e.
$
G \subseteq \mathcal{H}^k
$
In other words, each element of $G$ is constructed  by sampling a sequence of $k$ elementary hash functions $h_1, \ldots, h_k$ uniformly at random from $\mathcal{H}$ and then setting $g: \{0, 1\}^d \to \{0, 1\}^k$, $g(x) = (h_1(x), \ldots, h_k(x))$ to be their concatenation. That is, $\forall p \in \{0, 1\}^d$, $g(p) = (h_1(p), \ldots, h_k(p))$. 
%The values $L$ and $k$ are specified later in \cref{def:lsh_params}. 
The intuition is that concatenating hash functions from $\mathcal{H}$ allows to exponentially increase the gap between probabilities $p_1$ and $p_2$
in such a way that if we sample $L$ hash functions then two
points $p,q$ with distance at most $r$ are likely to have at 
least one $g_i$ with $g_i(p) = g_i(q)$ and we can recover the point. At the same time, the number of points with $\dist(p,q) > cr$ and $g_i(p) = g_i(q)$ will be small, so that the recovery can be done efficiently.

The proof of \cite{IndykM98} essentially follows by choosing the parameters $k$ and $L$ appropriately, which is done as follows. We set
$p_1 = 1-r/d$ and $p_2 = 1-cr/d$ such that $\mathcal H$
is $(r,cr,p_1,p_2)$-sensitive. We define $\rho = \frac{\log 1/p_1}{\log 1/p_2}$ and $\ell = n^{\rho}$. For an input set $P\subseteq \{0,1\}^d$ of $n$ points we set $k = \lceil \log_{1/p_2} n \rceil$. For this setting of parameter, we get the guarantees provided in the above Theorem. In order to get a smaller failure probability we set $L=\lceil\lambda \ell \rceil$, which results in a failure probability of $1-\exp({-\Omega(\lambda)})$.

One particular property of the above LSH construction is that
there is a chance of incorrectly returning that there is no near point (i.e. returning $\bot$). Such an incorrect answer will be called a \emph{false negative}.

\begin{definition}[False negative]
Consider an LSH data structure initialized for the point set $P\subseteq \{0,1\}^d$. We say that the answer to a query $q \in \{0, 1\}^d$ is a \textbf{false negative} query for this LSH instance, if $\lsh(q, G) = \bot$ but there exists a point $p \in P$ such that $\dist(p, q) \leq r$.
\end{definition}

% \begin{algorithm}[h]
% \caption{Construction of LSH hash function set.} \label{alg:lsh_constr}
% \KwIn{A set of $n$ points $P$, parameters $cr$, $L$, $k$ according to \cref{def:lsh_params} and $\mathcal{H}$ to \cref{eq:def_H}.}
% Let $G$ be an empty multiset\;
% \For{$i = 1 \ldots L$}{
%     Sample $h_1, \ldots, h_k$ uniformly at random from $\mathcal{H}$\;
%     $g_i \gets (h_1, \ldots, h_k)$ \tcp*[l]{Concatenation of $h_1, \ldots, h_k$}
%     $G \gets G \cup \{ g_i \}$\;
% }
% \Return an \lsh{} data structure initialized with parameter $cr$ and hash functions $G$\;
% \end{algorithm}

\textbf{Defining the adversary.}
The goal of this paper is to study how one can force LSH to return a false negative by doing adaptive queries. We take the role as an \emph{adversary} and our goal is design an \emph{adversarial algorithm} that has access to an LSH instance and quickly forces it to return a false negative. In particular, we provide guarantees on the number of queries required until LSH makes a false negative.

In the following, we will describe our formal setting. Our LSH instance
is constructed according to the scheme described in Section \ref{sec:basics}.
Recall that this LSH scheme samples a multiset $G$ of $L$ hash functions $g_1,\dots, g_L$ uniformly at random from the set of allowed hash functions of the form $g:\{0,1\}^d \rightarrow \{0,1\}^k$ with $g(p) = (h_1(p),\dots, h_k(p))$ for elementary hash function $h_j$.  
The LSH scheme receives as input a point set $P\subseteq \{0,1\}^d$ of size $n$ and parameters $n,d,c,r$ and $\lambda$ and calculates $\ell, \rho, L$ and $k$ as described earlier.
Our adversarial algorithm may use the input point set $P$, its size $n$ and dimension $d$, and the parameters $r,c$ and $\lambda$ of the LSH scheme and can interact with the LSH data structure using the procedure \ref{alg:lsh_query} described below. 
That is, the LSH data structure is in a black box and the only interaction 
uses procedure  \ref{alg:lsh_query}. In particular, the randomness used to construct the LSH instance and so the hash functions is unknown to the adversarial algorithm.

\begin{algorithm}
\begin{algorithmic}
\caption{QueryLSH($q, r, c$)} \label{alg:lsh_query}
%\KwData{Parameter $cr$ according to \cref{def:lsh_params} and a hash function set $G$ made available during initialization.}
\STATE $S \gets \emptyset$
\FORALL{$g \in G$}
    \STATE $S=S \; \cup  \{p\in P: g(p) = g(q) \text{ and } \dist(p,q) \le cr \}$
\ENDFOR
\RETURN any point in $S$, or $\bot$ if $S$ is empty
\end{algorithmic}
\end{algorithm}
Note that in an LSH implementation one would store all buckets of $g$ in a hash table. This way, line 3 can be implemented without going through all points. 
We also remark that since we do not have control of the parameters and point set our adversarial algorithm cannot be expected to work in all settings (for example, if the number of hash functions is very large, then no false negatives exist).

\section{Adversarial Algorithm}
\label{sec:alg_analysis}

Define the set 
$
\coll(p, q)=\{g \in G: g(p)=g(q)\}$ to be the set of hash functions in $G$ that collide on $p$ and $q$.
Our algorithm starts with finding a point $z$ in the dataset such that any other point is at least at a distance of $2cr$ from $z$. We assume that such a point exists, and we experimentally explore the case when it doesn't in \cref{sec:experim}. We also show that this assumption is reasonable for a randomly chosen point set in \cref{lem:rand_point_sep}. We refer to $z$ as \textbf{origin point}. 
% We will also require a number of inequalities about the input parameters to hold. \Cref{lem:params} shows that all of them are satisfied for a wide range of parameters.

\textbf{Algorithm outline.}
The idea of the algorithm is to query a random point $q$ at distance $r - t$ such that the expected number of hash functions hashing $q$ to the same bucket as $z$ is at most $t/2$. We show that this is true  when $t=2e^2 (\lambda + 1)$, which is the value we use throughout this section. This implies that with a constant probability $|\coll(q,z)| \leq t$. The algorithm now does $t$ iterations and in each of them it moves point $q$ away from $z$ by flipping one bit in it in such a way that the size of the set $\coll(q,z)$ also decreases by at least $1$. 

To find this bit, the algorithm flips bits in $q$ at $cr - \dist(q, z)$ randomly chosen positions among those where $q$ is equal to $z$ and obtains a point $\tilde{q}$. $\tilde{q}$ is located at distance $cr$ from $z$, so with a high probability $\lsh(\tilde{q}) = \bot$. Now consider a \textbf{path} between $q$ and $\tilde{q}$, i.e. a sequence of points starting with $q$ and ending with $\tilde{q}$, where each two consecutive points differ in one bit and where each point is further away from $z$ than its predecessor. Because $\tilde{q}$ doesn't hash with $z$ but $q$ does, there are two consecutive points $q'$, $q''$ on this path with the same property. Consider the bit where they differ. Because $\coll(q', z) \neq \emptyset$ and $\coll(q'', z) = \emptyset$, this bit must lie in the support of all functions in $\coll(q', z)$.\footnote{
for an elementary hash function $h^{(i)} \in \mathcal{H}$ its \textbf{support} is defined as the set $\{i\}$.
For a hash function $g \in G$, $g = (h_1, \ldots, h_k)$ its \textbf{support} is the union of supports of $h_1, \ldots, h_k$.}
Therefore, by flipping it in $q$ the algorithm removes these functions from $\coll(q, z)$.
It is easy to see that at the end we get a point $q$ at a distance at most $r$ from $z$ such that $\coll(q,z) = \emptyset$, so $\lsh(q) = \bot$ and $q$ is a false negative.

%\subsection{Key lemmas}
%\label{sec:key_lemmas}
%In this section
In the following we
prove \Cref{lem:near_coll,lem:far_coll}, which are the key results underpinning the correctness of our algorithm. \Cref{lem:near_coll} essentially shows that a point $q$ sampled at distance $r - t$ collides with $z$ in at most $O(t)$ hashings with constant probability. \Cref{lem:far_coll} that says that if we move this point randomly and monotonically away from $z$ to distance $cr$, it will not hash at all with $z$. 

A key property of these lemmas is that, unlike in the standard analysis of LSH, we fix the hash functions and sample points at random. This gives us a different perspective on how LSH behaves when queried with random points, which reflects the fact that our goal is to construct an adversary for LSH.

First, we state some intermediate results bounding the size of $k$ and the support of hash functions. Their proof is given in \cref{sec:appendix}.

\begin{restatable}[Bounds on $k$]{lemma}{kbounds} \label{cor:k_bounds}
If $cr/d \leq 1/5$ and $n > e$, one has
$2 \frac{d}{cr} \ln n \geq k \geq 0.5 \frac{d}{cr} \ln n$.

\end{restatable}
\begin{algorithm}[h]
\caption{Adaptive adversarial walk from $z$. } \label{alg:simple}
\begin{algorithmic}[1]
\REQUIRE{Origin point $z$, parameters $r$, $c$, and $\lambda$}
\STATE $t \gets 2 e^2 (\lambda + 1) $
\STATE Sample $q$ uniformly at random from all points at distance $r - t$ from $z$\label{line:sample_point}
\WHILE{$\lsh(q) = z$} \label{line:outer_loop}
    \IF{$\dist(q, z) \geq r$}
    \RETURN $\bot$
    \ENDIF
    \STATE $q' \gets q$ \label{line:slow_inner_start}
    \WHILE{$\lsh(q') = z$} \label{line:inner_loop}
        \STATE Let $I = \{i\in [d]: q_i' = z_i\}$
        \STATE Select $j\in I$ uniformly at random 
        \STATE Set $q_j' = 1 - q_j'$
        \IF{$\dist(q', z) > cr$}
            % $q' \gets q$\;
            \RETURN $\bot$
        \ENDIF
        \label{line:slow_inner_end}
    \ENDWHILE
    \STATE Set $q_j = 1 - q_j$
\ENDWHILE
\RETURN $q$
\end{algorithmic}
\end{algorithm}

\begin{restatable}[Support lower bound]{lemma}{supplb} \label{lem:supp_lb}
    Let $\lambda \le n$ and $n > e$. Then with probability at least $1 - \frac{1}{n}$ the size of the support of each of the functions $g \in G$ is at least $k - 7 \ln n \cdot \max\{1, \frac{k^2}{2d}\}$.
\end{restatable}

The next two lemmas lower bound the size of the support of hash functions in LSH when they are chosen as described in Section \ref{sec:prelims}. 
We now show the first lemma, which bounds the number of hash functions hashing a point at distance $r - t$ from $z$ in the same bucket with $z$.
We outline the proof below, and present the full proof in \cref{app:sec41}. 

\begin{restatable}{lemma}{nearcoll} \label{lem:near_coll}
Let $n > e$, $8e^2 (\lambda + 1) \ln n \leq cr$,  $28 \ln^3 n / c^2 \leq r \leq d/(14 \ln n)$,  $r \leq d/(5c)$ and $c < \ln n$. Let the hash functions $G$ be such that the size of the support of each is at least $k - 7 \ln n \cdot \max\{1, \frac{k^2}{2d}\}$.
Let $t =  2 e^2 (\lambda + 1)$.
Let $q$ be a point at distance $r - t$ from $z$ chosen uniformly at random. The expected size of $\coll(q, z)$ is at most $e^2 (\lambda + 1)$,  where the expectation is taken with respect to the randomness of $q$.
\end{restatable}

\begin{proof}[Proof outline]
    We fix an arbitrary set $G$ of $L$ hash functions that satisfy the condition of the lemma. 
    By linearity of expectation, 
$
\E [|\coll(q,z)|] = \sum_{g\in G} \E[X_g] = \sum_{g\in G} \Pr[g(q) = q(z)],
$
where $X_g$ is the indicator random variable for the event $g(q)=g(z)$ and the randomness is over the choice of $q$ as in the lemma.
    In the following we consider an arbitrary $g\in G$
    and derive an upper bound for the collision probability
    of $q$ and $z$.
    Indeed, let $x = 7 \ln n \cdot \max\{1, \frac{k^2}{2d}\}$.
    We first observe (see \cref{app:sec41} for more details) that
    $\label{eq:0wireg90hweg9_short}
    \Pr[g(q) = g(z)] \leq \left(1 - \frac{k - x}{d}\right)^{r- t}.
    $
    The basic argument is that sampling a point at distance $r$ from $z$ is not worse than just picking $r$ bits to flip in $z$ independently at random, and then bounding the probability that none of those chosen bits lie in the support of $g$.

    Next we can use this upper bound to argue that
    $
    \sum_{g\in G} \Pr[g(q) = g(z)] \le L \cdot \left(1 - \frac{k - x}{d}\right)^{r- t}.
    $
    Using $L = \lceil \lambda \ell \rceil \leq (\lambda + 1) \ell$ and
    %\begin{equation*}
    \(
        \ell = n^{\rho} = \left(1 - \frac{r}{d}\right)^{-\log_{1/p_2} n} \leq \left(1 - \frac{r}{d}\right)^{-k}
    \)
    %\end{equation*}
    we obtain
    $
    \E_{q}\left[ |\coll(q, z)| \right] \leq (\lambda + 1) \left(1 - \frac{r}{d}\right)^{-k} \left(1 - \frac{k - x}{d}\right)^{r -t} .
    $
    It remains to show that the r.h.s. is at most $e^2(\lambda + 1)$.
    Because $x$ and $t$ are considerably smaller than $k$ and $r$ respectively, this essentially boils down to show that
    $ \label{eq:outline_core}
    \left(1 - \frac{r}{d}\right)^{-k} \cdot \left(1 - \frac{k}{d}\right)^{r} \leq O(1).
    $
    This is simple to bound using the Taylor expansion of both terms (see \cref{app:sec41} for detailed calculations).
\end{proof}

We now prove the second lemma, which bounds the probability that a point at distance $cr$ will be hashed with the origin.

\begin{lemma} \label{lem:far_coll}
    Let $n > e$, $cr/d \leq 1/5$, $c \geq 1 + \frac{80 + 16 \ln(\lambda + 1)}{\ln n}$, $8e^2(\lambda + 1) \leq n^{1/8}$. Let $G$ be such that no hash functions $g \in G$ has support smaller than $k/2$. Let $q$ be a point at distance at most $r$ from $z$ such that $|\coll(q, z)|$ is at most $2e^2(\lambda + 1)$. Let $q'$ be a point chosen uniformly at random among all of the points at distance $cr$ from $z$ such that for all $i \in [d]$, if $q_i \neq z_i$, then $q'_i \neq z_i$. Then with probability at least $1 - 1/(32 e^2(\lambda + 1))$, $\coll(q', z) = \emptyset$, i.e. no function $g$ hashes $q'$ and $z$ to the same bucket.
\end{lemma}
\begin{proof}
    Fix $g \in \coll(q, z)$. %We use the same proof technique as in \cref{lem:near_coll}.
    First notice that $q'$ defined in lemma statement can be generated as follows. One first selects a uniformly random subset $S$ of $cr - \dist(q, z)$ coordinates in $I = \{ i \in [d]: q_i = z_i \}$. Then one lets $q'_i=q_i\oplus 1$ (negation of $q_i$) for $i\in S$ and $q'_i=q_i$ for $i\in [d]\setminus S$. We denote this distribution by $\mathcal{D}(q)$.  Note that because $g(q) = g(z)$ (since $g\in \coll(q, z)$ by assumption), the support of $g$ is a subset of $I$. 
    We will show that     
    $ %\label{eq:far_coll_1}
        \Pr_{q' \sim \mathcal{D}(q)}[g(q') = g(z)] \leq \left(1 - \frac{k}{2d}\right)^{(c-1)r}.
    $
    In order to establish the above, 
    %similarly to the proof of \cref{lem:near_coll}, 
    we view the generation of $q'$ as a two step process. First we uniformly sample $cr - \dist(q, z) \geq cr - r$ bits from $I$ {\em with replacement} into a set $B$ (removing duplicates) and flip them in $q$ -- denote the result by $\tilde q$. Next, we choose $cr - \dist(\tilde q, z)$ more bits from $I \setminus B$ and flip them in $\tilde q$ to get $q'$.
    It again holds that $g(\tilde q) \neq g(z)$ implies $g(q') \neq g(z)$, so $\Pr_{q'}(g(q') \neq g(z)) \geq \Pr_{\tilde q}(g(\tilde q) \neq g(z))$.

    The probability that each bit chosen in the first step lies in support of $g$ is at least $\frac{k/2}{d - \dist(q, r)} \geq \frac{k}{2d}$. Hence, the probability that none of them lie in the support, i.e. $\Pr_{\tilde q}(g(\tilde  q) = g(z))$, is at most $(1 - \frac{k}{2d})^{cr - r}$. This establishes the inequality.
    %\eqref{eq:far_coll_1}.
    
    By \cref{cor:k_bounds}, $k \geq 0.5 \frac{d}{cr} \ln$.
    By applying \cref{lem:exp_bounds} to $(1 - \frac{k}{2d})^{cr - r}$, we obtain:
    %\begin{equation}\label{eq:08whgf8h9fasf}
    %\begin{split}
    %\Pr[g(q') = g(z)]  &\leq \left(1 - \frac{k}{2d}\right)^{(c-1)r}\\ &\leq \exp\left[- \frac{k(c-1)r}{2d} \right] \\ 
    %&=\exp\left[-\frac{d \ln n}{2cr} \cdot \frac{(c- 1)r}{2d}\right]\\
    %&=\exp\left[-\frac{\ln n}{2c} \cdot \frac{(c- 1)}{2}\right]\\
    %&\leq n^{-\frac{c-1}{4c}}.
    %\end{split}
    %\end{equation}
    $
    \Pr[g(q') = g(z)]  \le \left(1 - \frac{k}{2d}\right)^{(c-1)r} 
    \leq  n^{-\frac{c-1}{4c}}.%\label{eq:08whgf8h9fasf}
    $
    Now consider two cases. First, if $c \geq 2$, then the rhs of the above inequality satisfies $n^{-\frac{c-1}{4c}} \leq n^{-1/4} \leq \frac{1}{(8e^2(\lambda + 1))^2}$
    since $8e^2(\lambda + 1) \leq n^{1/8}$ by assumption of the lemma. Otherwise, using that $c \geq 1 + \frac{80 + 16 \ln(\lambda + 1)}{\ln n}$ by the assumption of the lemma, we get
    $\frac{c-1}{4c} \geq \frac{c - 1}{8} \geq \frac{10 + 2 \ln(\lambda + 1)}{\ln n} \geq \ln(8e^2(\lambda + 1))/\ln n$.
    Thus, in both cases one has $n^{-\frac{c-1}{4c}} \leq \frac{1}{(8e^2(\lambda + 1))^2}$. Finally, since $|\coll(q,z)| \leq 2e^2 (\lambda + 1)$ by the assumption of the lemma, by union bound across all $g\in \coll(q, z)$, $\Pr[\coll(q', z) \neq \emptyset] \leq 1/(32 e^2(\lambda + 1))$, as required.
\end{proof}

\subsection{Analysis of the algorithm} \label{sec:simple}

Equipped with~\cref{lem:near_coll} and~\cref{lem:far_coll}, we are now ready to analyse~\cref{alg:simple}. Throughout this section we make a strong assumption that the input point $z$ is \textbf{isolated}, that is, the distance from $z$ to any other point in $P$ that is not equal to $z$ is at least $2cr$. This means that as long as we only make queries $q$ at distance at most $cr$ to $z$, the only two possible results of $\lsh$ are $z$ and $\bot$. This in turn means that $\lsh(q) = z$ iff $\coll(q, z) \neq \emptyset$.

First, we analyze the inner loop of the algorithm --- lines~\ref{line:slow_inner_start} to \ref{line:slow_inner_end}. Its goal is to find a bit to flip in $q$ such that the size of $\coll(q,z)$ would decrease afterward. Of course, our algorithm only gets black-box access to \lsh, and cannot test the size of $|\coll(q, z)|$ in general. However, for a point $q'$ one can, using black-box access to \lsh, distinguish between $|\coll(q', z)|=0$ and $|\coll(q', z)|>0$: due to the assumption that $z$ is an isolated point we have $|\coll(q', z)|>0$ iff $\lsh(q', z)\neq \bot$. As we show, this ability suffices. \cref{alg:simple} simply keeps flipping  bits in $q$ at random (but restricted to coordinates where $q$ and $z$ agree) until $z$ is no longer returned by \lsh~ on the modified query: the last flipped bit must lead to a decrease in $|\coll(q, z)|$ even in the original query $q$!  More formally,  our algorithm finds two points $q'$ and $q''$ such that $q'$ lies on a path\footnote{Recall that by a \textbf{path} between $a$ and $b$ we mean a sequence of points starting with $a$ and ending with $b$, where each two consecutive points differ in one bit and where each point is further away from $a$ than its predecessor.} between $q''$ and $q$, and $\coll(q'', z) = \emptyset$ and $\coll(q', z) \neq \emptyset$, and that differ only in one bit. This is enough to conclude that this bit lies in support of all hash functions of $\coll(q',z)$, so by flipping it in $q$ we would remove all of those hash functions from $\coll(q,z)$.

\begin{lemma}[Inner loop of the adversarial walk] \label{lem:simple_inner_loop}
    Let $n > e$, $cr/d \leq 1/5$, $c \geq 1 + \frac{80 + 16 \ln(\lambda + 1)}{\ln n}$, $8e^2(\lambda + 1) \leq n^{1/8}$. Let $G$ be such that no hash functions $g \in G$ has support smaller than $k/2$. Suppose that the point $z$  passed to \cref{alg:simple} is located at distance at least $2cr$ from any other point in $P$. Let $q$ be sampled in line~\ref{line:sample_point} be such that $|\coll(q, z)| \leq 2 e^2 (\lambda + 1)$.
    For a fixed iteration of the outer loop in line~\ref{line:outer_loop} of \cref{alg:simple}, with probability at least \(1 - 1/(32e^2 (\lambda + 1))\) 
    the inner loop between lines~\ref{line:slow_inner_start} and \ref{line:slow_inner_end} finds two points $q'$ and $q''$ such that $\lsh(q') \neq z$, $\lsh(q'') = z$, and $q'$ and $q''$ differ in one bit.
    % The inner loops takes time at most $O(cr - r + t)$.
\end{lemma}
\begin{proof}
    Fix an iteration of the outer loop. Then the inner loop at line~\ref{line:inner_loop} essentially does a random walk away from $z$, always increasing the distance from $z$ after every flip. There are two possibilities: either it finds a point at distance at most $cr$ from $z$ such that $\lsh(q) \neq z$, i.e. $\coll(q', z) = \emptyset$, or it reaches a point $q''$ at distance $cr$ such that $\lsh(q'') = z$, where the algorithm would return $\bot$.
    
    Assume that the latter happens. Since the bits flipped are chosen at random, this process is equivalent to flipping $cr - \dist(q, z)$ different bits of $q$ where $q$ is not equal to $z$. This is the same process as in statement of \cref{lem:far_coll}, so by this lemma the probability that this case takes place is at most \(1/(32e^2 (\lambda + 1))\).
\end{proof}

We are now ready to analyze \cref{alg:simple}. 
Here we present only the conceptual part of the proof of \cref{thm:simple}. 
The full proof can be found in \cref{app:simple_analysis}

\begin{restatable}{theorem}{combined} \label{thm:combined}
Let $n>e$, $28 \ln^3 n \leq r \leq d/(28 \ln n)$, $\lambda \leq \frac{1}{8e^2} \min\{\frac{r}{\ln n}, n^{1/8}\} - 1$, and $1 + \frac{80 + 16 \ln(\lambda + 1)}{\ln n} \leq c \leq \ln n$.
%
%Let $n > e$ and
%\begin{itemize}
%    \item $28 \ln^3 n \leq r \leq d/(28 \ln n)$,
%    \item $\lambda \leq \frac{1}{8e^2} \min\{\frac{r}{\ln n}, n^{1/8}\} %- 1$,
%    \item $1 + \frac{80 + 16 \ln(\lambda + 1)}{\ln n} \leq c \leq \ln n$.
%\end{itemize}
Let $t = 2e^2 (\lambda + 1)$ be the parameter from \cref{alg:simple} (\cref{alg:fast}). Suppose that the point $z$  passed to \cref{alg:simple} 
(\cref{alg:fast})
 is located at a distance of at least $2cr$ from any other point in $P$. With probability at least $1/4 - 1/n$ the algorithm finds a point $q$ at a distance of at most $r$ from $z$ such that querying LSH with $q$ returns no point.
It uses at most $O(cr\cdot \lambda)$ ($O(\log (cr) \cdot \lambda)$ for 
\cref{alg:fast}) queries to the LSH data structure.
\end{restatable}

\begin{proof}[Proof outline]
We start with \cref{alg:simple}.
Let $q^{(j)}$ be the value of $q$ at the end of the iteration $j$ inner loop of \cref{lem:simple_inner_loop}, and let $q^{(0)}$ be its value at the beginning of the first iteration of this loop. Let $j^*$ be the number of iterations.

For each iteration $j$ of the inner loop, by \cref{lem:simple_inner_loop} two points $q'$ and $q''$ are found. The bit where they differ must lie in the support of one of the hash functions in $\coll(q'', z)$ since $\lsh(q') \neq z$. Since $q^{(j)}$ and $q^{(j-1)}$ differ in only this bit, and because $\coll(q', z) \subseteq \coll(q^{(j-1)}, z)$, it holds that $|\coll(q^{(j)}, z)| < |\coll(q^{(j-1)}, z)|$, so the size of the $|\coll(q^{(j)}, z)|$ decreases by at least $1$ after each execution of the loop.

By \cref{lem:near_coll} and Markov's inequality, $|\coll(q^{(0)}, z)|$ is at most $2e^2(\lambda + 1)$ with constant probability, so at most $2e^2 (\lambda + 1)$ iterations will be made until $\coll(q^{(j^*)}, z) = \emptyset$, which also means that $\lsh(q^{(j^*)}) = \bot$ so the loop stops. On the other hand, $\dist(q^{(j^*)}, z) \leq r$, which makes $q^{(j^*)}$ a false negative query. 
% \paragraph{Discussion of parameter bounds}
% \xxx[MM]{TODO: finish}
% \[
% \lambda \leq \frac{1}{8e^2} \min\left\{\frac{r}{\ln n}, n^{1/8}\right\} - 1
% \]
% This is essentially two upper bounds on $\lambda$, in terms of $r$ and $n$.
% Note that by union bound, if $\lambda \geq O(r \log d \poly \log n) \geq \log \binom{d}{r} \cdot O(\poly \log n)$, with high probability the structure correctly answers all queries within the distance $r$ of any point in the dataset, meaning that it never produces a wrong answer. Hence our upper bound in $r$ is tight up to polylogarithmic factors.
% The upper bound of $O(n^{1/8})$ well covers the typical range of $\lambda$, as it is usually set to be $O(\poly \log n)$.

% \xxx[MM]{Not sure what to say about the rest.}

With a simple modification of the inner loop of \cref{alg:simple}, we can exponentially reduce the dependence on $cr$ in the running time and the number of queries. 
% Since the algorithm gets only one bit of information from each query, this means it only requires the amount of information roughly equal to the size of $\lambda$ indices encoding a position in a $d$-dimensional vector to construct a failing query. \xxx[MK]{I am unable to parse this last sentence.}
% 
We state the algorithm and give the proof of its correctness in \cref{sec:fast_analysis}. At a high level, instead of using a linear search as in \cref{alg:simple}, to find a pair of points $q''$, $q'$ such that $\coll(q'', z) \neq \emptyset$ and $\coll(q', z) = \emptyset$, we use a binary search in the inner loop. We still start with the same point $q$. However, we now flip the $cr - \dist(q,z)$ bits in $q$ in positions where $q$ equals $z$ to get $q^{right}$. Then by \cref{lem:far_coll}, $q^{right}$ does not hash with $z$. Now we use a binary search on the path from $q$ to $q^{right}$ to find the aforementioned pair of points. The analysis is analogous to that of \cref{lem:simple_inner_loop}.
\end{proof}

\section{Discussion of potential defenses against our attacking scheme}

As was mentioned in \cref{sec:related}, there already exists large literature on the robustification of LSH, many of which work against our adversary. In particular, we highlight \cite{Ahle17}, which provides a false-negative-free Las-Vegas type data structure with query time $dn^{1/c + o(1)}$ and used space $dn^{1 + 1/c + o(1)}$, achieving the same asymptotic performance as \cite{IndykM98}, as well as \cite{CherapanamjeriN20}, which can be applied to LSH to guarantee perfect robustness by paying $O(d)$ times more space.

However, both approaches are not perfect within the same parameter bounds, as approach of \cite{Ahle17} doesn't give any guarantees on query time in this setting, and the space usage of \cite{CherapanamjeriN20} is too prohibitive in settings with large $d$. We experimentally investigate how the latter, and a construction based on differential privacy, fares against our adaptive adversary when the available space is limited in Experiment 6 in \cref{app:exp}. Interestingly, these defences perform well in this setting, which suggests that they might be reasonably effective even against a general adversary.
\section{Experiments}
\label{sec:experim}
The goal of this section is to verify the efficacy of our adversary on real datasets, where an isolated point does not necessarily exist, and for parameter configurations outside of the guarantees of \cref{thm:simple,thm:fast}.

\begin{figure*}[!t]
    \centering
    \includegraphics[width=\textwidth]{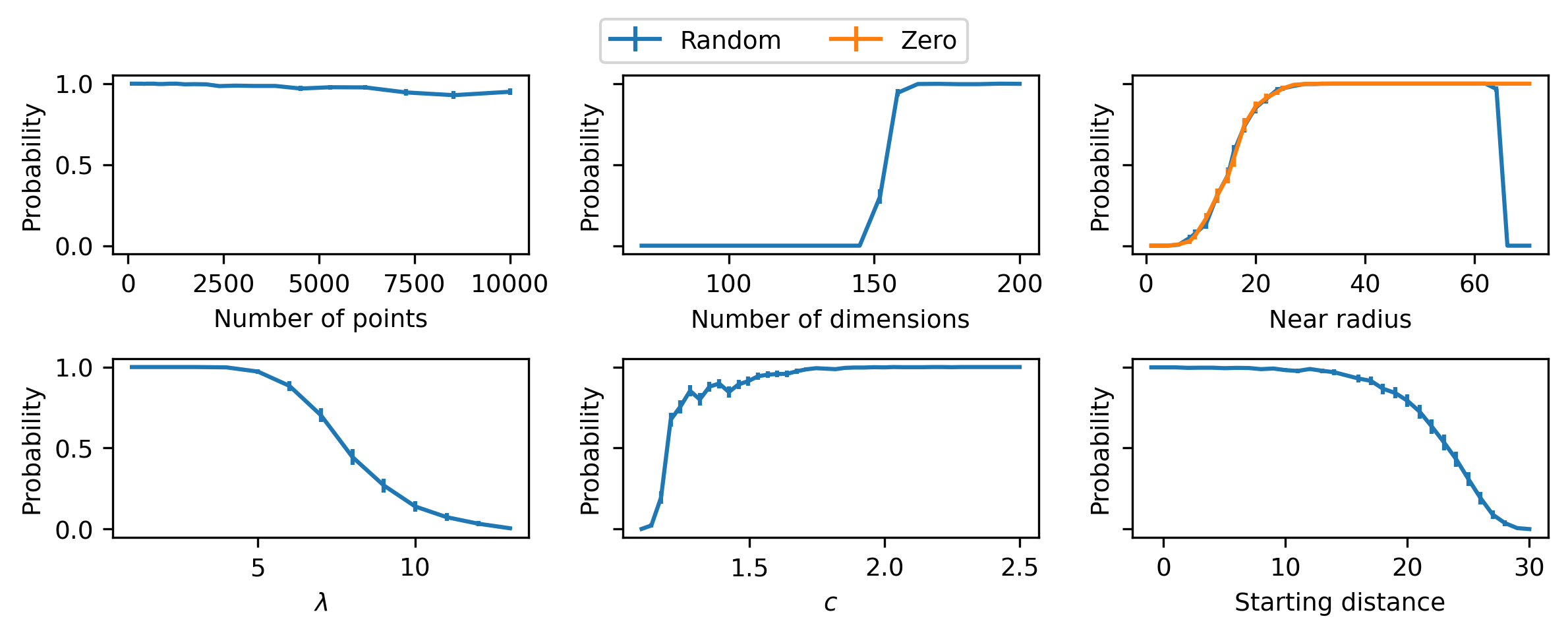}
    \caption{Dependence of success probability on various parameters. All experiments are done on the Random dataset, with the third one also featuring the Zero dataset.}
    \label{fig:param-variation}
\end{figure*}

\begin{figure*}[t!]
    \centering
    \includegraphics[width=\textwidth]{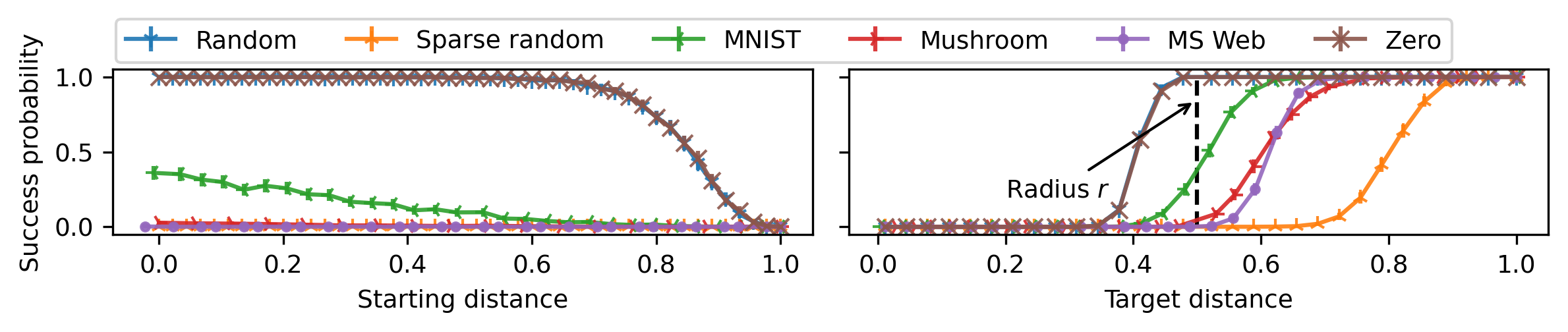}
    \caption{Dependence of the success probability on the value of $t$ and on the value of desired distance of false negative query from the origin. }
    \label{fig:all-data}
\end{figure*}

%\subsection{Details of experimental setup}
%\label{sec:exp_setup_details}
\textbf{Implementation details.}
Our implementation of LSH follows \cref{alg:lsh_query} in \cref{sec:prelims} \cite{IndykM98}. To summarise, given a query point $q$, we iterate through all points in the dataset that hash together with $q$ under some hash function, then return the first one which lies at distance at most $cr$ from $q$, and return nothing if there is no such point.
The adversary is implemented according to \cref{alg:fast}.

\textbf{Experiment 1: Influence of parameters on success probability.}
The goal of this experiment is to explore within which parameter setting the algorithm works. 
For this experiment we use synthetic datasets, since for them we can easily vary the number of points and dimensions. We use two datasets:

\textit{Zero}. This dataset entirely consists of $n$ copies of the all-zero point. This dataset is the easiest for our algorithm to find a false negative in, as it is equivalent to a dataset containing only a single point, which would automatically be isolated.

\textit{Random}. Each point is sampled i.i.d. uniformly, with each bit having a $1/2$ probability to be set to $1$.

For each point on the plots, the experiments were run $1000$ times, and the result is the mean of the associated values.
Plots contain error bars computed as standard error of the mean, but sometimes they are too small to be visible.

The parameters that we vary are $n$, $d$, $r$, $\lambda$, $c$ and $t$. Each parameter is changed individually. The default setting of parameters is $n=1000$, $d=300$, $\lambda = 4$, $r=30$, $c=2$.
The origin point $z$ is picked at random.
The offset $t$ is set to $r$, or, in other words, the initial point $q$ is initialized to be equal to $z$. This is different than what was used in \cref{alg:simple,alg:fast}. However, the intuition is that if instead of setting $q$ to be a random point at distance $2e^2(\lambda + 1)$ from $z$ we set $q=z$ and then move it away from $z$ in a way that ensures that the size of $\coll(q, z)$ decreases for every changed bit, the probability that the algorithm will find a false negative query would only increase. This intuition ends up being supported by the experiment results.

The results are presented in \cref{fig:param-variation}.  
The sixth subplot, counting left-to-right and top-to-bottom, supports the claim that we just made: the lower the starting distance of the point $q$, which is equal to $r - t$, the higher the success probability of the algorithm.
The first and fifth subplots are straightforward. The number of points $n$ has a very weak effect on the success probability. Parameter $c$ also has a weak influence if it is large enough, and when it is small it is unlikely that a point at distance $cr$ does not hash with $z$, see \cref{lem:far_coll}. 
The drop in probability on the left-hand side of the third subplot where the near radius $r$ is varied and the right-hand side of the fourth where $\lambda$ is varied can be explained by the fact that a random point at distance $r$ has in expectation $\Omega(\lambda)$ hash functions hashing it with $z$. Hence, intuitively, for it to be possible to eliminate all collisions with $z$, it must be that $t \geq \Omega(\lambda)$, which also implies that $r \geq \Omega(\lambda)$ is necessary.
% This dependence is better seen in \cref{fig:radius-lambda}, where both $\lambda$ and $r$ are being varied at the same time. Note that $r$ must be approximately at least $4\lambda$ for the success probability to be high.
% The drop on the second graph for low dimensions and the third graph for high radius for the Random dataset is explained by the fact that at those parameter settings, $2cr$ is approximately bigger than $d$, so no point in the Random dataset is isolated. This also explains why there is no drop for the Zero dataset. 
% This dependency can also be seen very well on \cref{fig:radius-dimension}, where both the dimension and the radius are varied. The blue line $r=0.15 d$ represents a setting where the algorithm succeeds with high probability. This is why we use it as a default value for $r$.

\textbf{Experiment 2: Efficiency across datasets.}
In this experiment we measure how well our algorithm performs on a variety of different datasets. We use the datasets from Experiment 1, as well as one additional synthetic dataset and 3 datasets with real data.
The new synthetic dataset is Sparse random. Each point there is generated by independently sampling each bit to be $1$ with probability $1/15$, and to be $0$ otherwise.
The remaining datasets are:
%\begin{itemize}
    %\item 
 We use embeddings of the data sets \textit{MNIST} \cite{lecun1998mnist}, \textit{Anonymous Microsoft Web Data }(MSWeb) \cite{misc_anonymous_microsoft_web_data_4} and \textit{Mushroom} \cite{misc_mushroom_73,wulff1982audubon} 
 into Hamming space. For MNIST we flatten the images into vectors whose Boolean variables indicate the presence of a pixel. For MSWeb  and Mushroom we use one-hot encoding to map the data into Hamming space. A more detailed description can be found in the \cref{app:exp}. 
%Those datasets are fixed in all of the experiments.

Due to computational constraints, we choose to limit the size of each dataset to at most $10000$. For MNIST and MSWeb we only use the first $10000$ points. We use all $8124$ points in Mushroom. For the synthetic datasets we generate $10000$ points with the dimension $300$, which is about the same order of magnitude as the dimensions of real datasets.
We choose a setting of parameters with which our algorithm has a good chance of success as per the results of Experiment~1. As the number of dimensions varies depending on the dataset, we set $r = 0.15 d$. $c = 2$, $\lambda = 4$ and the starting distance is $0$ on the second subplot.
Otherwise, we use the same setup as Experiment 1, except that for fixed datasets the origin point is the first one in the dataset.
On the first subplot the starting distance is measured as a fraction of $r$.
On the second subplot the goal of the algorithm is to find a query to the LSH data structure that returns $\bot$ at a distance from the origin specified by the $x$ axis as a fraction of the $cr$. The dashed black vertical line represents the radius $r$ and is located at $0.5$ since $c=2$.
%
% Our analysis guarantees that when the adaptive algorithm starts at a random vertex at distance $r-t$ for a suitable choice of $t$ it finds a false negative query with a constant probability. In our first experiment, we want to understand how the choice of $t$ influences the success probability. 
% For this purpose, we let the starting distance from origin $r-t$ range from $0$ to $r$ and we empirically estimate the success probability of finding a false negative query.
%
The results in the first subplot of \cref{fig:all-data} show that success probability drops monotonically for all datasets as the starting distance grows. This is natural to expect given that similar behaviour was observed in Experiment 1. 
% What is more interesting is that most datasets have a low success probability regardless of the value of $t$.
%
% This is revealed in detail by the second subplot.
On the second subplot we can see that for the most datasets we can reliably find a negative query only at distance much larger than $r$. This, in a way, allows us to compare the ``hardness'' of datasets, as further away from the origin we move, the less functions collide the origin point $z$ and our current point, the less bits we need to change for the query to become negative.
As one can see, all of the real datasets lie between the Random and Sparse random on the second subplot. This %difficulty difference 
can be explained by the fact that the points in real datasets are much sparser than the Random, and the maximum distance from the most isolated point to the nearest point is smaller than the ideal $2cr$.

\bibliographystyle{plain}
\bibliography{main}

\begin{thebibliography}{10}

\bibitem{Ahle17}
Thomas~Dybdahl Ahle.
\newblock Optimal las vegas locality sensitive data structures.
\newblock In Chris Umans, editor, {\em 58th {IEEE} Annual Symposium on Foundations of Computer Science, {FOCS} 2017, Berkeley, CA, USA, October 15-17, 2017}, pages 938--949. {IEEE} Computer Society, 2017.

\bibitem{ABVT20}
Panagiotis Anagnostou, Petros Barbas, Aristidis~G. Vrahatis, and Sotiris~K. Tasoulis.
\newblock Approximate knn classification for biomedical data.
\newblock In {\em 2020 IEEE International Conference on Big Data (Big Data)}, pages 3602--3607, 2020.

\bibitem{AndoniI06}
Alexandr Andoni and Piotr Indyk.
\newblock Near-optimal hashing algorithms for approximate nearest neighbor in high dimensions.
\newblock In {\em 47th Annual {IEEE} Symposium on Foundations of Computer Science {(FOCS} 2006), 21-24 October 2006, Berkeley, California, USA, Proceedings}, pages 459--468. {IEEE} Computer Society, 2006.

\bibitem{AndoniILRS15}
Alexandr Andoni, Piotr Indyk, Thijs Laarhoven, Ilya~P. Razenshteyn, and Ludwig Schmidt.
\newblock Practical and optimal {LSH} for angular distance.
\newblock In Corinna Cortes, Neil~D. Lawrence, Daniel~D. Lee, Masashi Sugiyama, and Roman Garnett, editors, {\em Advances in Neural Information Processing Systems 28: Annual Conference on Neural Information Processing Systems 2015, December 7-12, 2015, Montreal, Quebec, Canada}, pages 1225--1233, 2015.

\bibitem{AndoniINR14}
Alexandr Andoni, Piotr Indyk, Huy~L. Nguyen, and Ilya~P. Razenshteyn.
\newblock Beyond locality-sensitive hashing.
\newblock In Chandra Chekuri, editor, {\em Proceedings of the Twenty-Fifth Annual {ACM-SIAM} Symposium on Discrete Algorithms, {SODA} 2014, Portland, Oregon, USA, January 5-7, 2014}, pages 1018--1028. {SIAM}, 2014.

\bibitem{AndoniLRW17}
Alexandr Andoni, Thijs Laarhoven, Ilya~P. Razenshteyn, and Erik Waingarten.
\newblock Optimal hashing-based time-space trade-offs for approximate near neighbors.
\newblock In Philip~N. Klein, editor, {\em Proceedings of the Twenty-Eighth Annual {ACM-SIAM} Symposium on Discrete Algorithms, {SODA} 2017, Barcelona, Spain, Hotel Porta Fira, January 16-19}, pages 47--66. {SIAM}, 2017.

\bibitem{AndoniR15}
Alexandr Andoni and Ilya~P. Razenshteyn.
\newblock Optimal data-dependent hashing for approximate near neighbors.
\newblock In Rocco~A. Servedio and Ronitt Rubinfeld, editors, {\em Proceedings of the Forty-Seventh Annual {ACM} on Symposium on Theory of Computing, {STOC} 2015, Portland, OR, USA, June 14-17, 2015}, pages 793--801. {ACM}, 2015.

\bibitem{AttiasCSS23}
Idan Attias, Edith Cohen, Moshe Shechner, and Uri Stemmer.
\newblock A framework for adversarial streaming via differential privacy and difference estimators.
\newblock In Yael~Tauman Kalai, editor, {\em 14th Innovations in Theoretical Computer Science Conference, {ITCS} 2023, January 10-13, 2023, MIT, Cambridge, Massachusetts, {USA}}, volume 251 of {\em LIPIcs}, pages 8:1--8:19. Schloss Dagstuhl - Leibniz-Zentrum f{\"{u}}r Informatik, 2023.

\bibitem{0001CP018}
Martin Aum{\"{u}}ller, Tobias Christiani, Rasmus Pagh, and Francesco Silvestri.
\newblock Distance-sensitive hashing.
\newblock In Jan~Van den Bussche and Marcelo Arenas, editors, {\em Proceedings of the 37th {ACM} {SIGMOD-SIGACT-SIGAI} Symposium on Principles of Database Systems, Houston, TX, USA, June 10-15, 2018}, pages 89--104. {ACM}, 2018.

\bibitem{BeimelKMNSS22}
Amos Beimel, Haim Kaplan, Yishay Mansour, Kobbi Nissim, Thatchaphol Saranurak, and Uri Stemmer.
\newblock Dynamic algorithms against an adaptive adversary: generic constructions and lower bounds.
\newblock In Stefano Leonardi and Anupam Gupta, editors, {\em {STOC} '22: 54th Annual {ACM} {SIGACT} Symposium on Theory of Computing, Rome, Italy, June 20 - 24, 2022}, pages 1671--1684. {ACM}, 2022.

\bibitem{Ben-EliezerJWY22}
Omri Ben{-}Eliezer, Rajesh Jayaram, David~P. Woodruff, and Eylon Yogev.
\newblock A framework for adversarially robust streaming algorithms.
\newblock {\em J. {ACM}}, 69(2):17:1--17:33, 2022.

\bibitem{misc_anonymous_microsoft_web_data_4}
Jack Breese, David Heckerman, and Carl Kadie.
\newblock {Anonymous Microsoft Web Data}.
\newblock UCI Machine Learning Repository, 1998.
\newblock {DOI}: https://doi.org/10.24432/C5VS3Q.

\bibitem{Charikar02}
Moses Charikar.
\newblock Similarity estimation techniques from rounding algorithms.
\newblock In John~H. Reif, editor, {\em Proceedings on 34th Annual {ACM} Symposium on Theory of Computing, May 19-21, 2002, Montr{\'{e}}al, Qu{\'{e}}bec, Canada}, pages 380--388. {ACM}, 2002.

\bibitem{CherapanamjeriN20}
Yeshwanth Cherapanamjeri and Jelani Nelson.
\newblock On adaptive distance estimation.
\newblock In Hugo Larochelle, Marc'Aurelio Ranzato, Raia Hadsell, Maria{-}Florina Balcan, and Hsuan{-}Tien Lin, editors, {\em Advances in Neural Information Processing Systems 33: Annual Conference on Neural Information Processing Systems 2020, NeurIPS 2020, December 6-12, 2020, virtual}, 2020.

\bibitem{Christiani17}
Tobias Christiani.
\newblock A framework for similarity search with space-time tradeoffs using locality-sensitive filtering.
\newblock In Philip~N. Klein, editor, {\em Proceedings of the Twenty-Eighth Annual {ACM-SIAM} Symposium on Discrete Algorithms, {SODA} 2017, Barcelona, Spain, Hotel Porta Fira, January 16-19}, pages 31--46. {SIAM}, 2017.

\bibitem{Cohen0NSSS22}
Edith Cohen, Xin Lyu, Jelani Nelson, Tam{\'{a}}s Sarl{\'{o}}s, Moshe Shechner, and Uri Stemmer.
\newblock On the robustness of countsketch to adaptive inputs.
\newblock In Kamalika Chaudhuri, Stefanie Jegelka, Le~Song, Csaba Szepesv{\'{a}}ri, Gang Niu, and Sivan Sabato, editors, {\em International Conference on Machine Learning, {ICML} 2022, 17-23 July 2022, Baltimore, Maryland, {USA}}, volume 162 of {\em Proceedings of Machine Learning Research}, pages 4112--4140. {PMLR}, 2022.

\bibitem{CohenNSS23}
Edith Cohen, Jelani Nelson, Tam{\'{a}}s Sarl{\'{o}}s, and Uri Stemmer.
\newblock Tricking the hashing trick: {A} tight lower bound on the robustness of countsketch to adaptive inputs.
\newblock In Brian Williams, Yiling Chen, and Jennifer Neville, editors, {\em Thirty-Seventh {AAAI} Conference on Artificial Intelligence, {AAAI} 2023, Thirty-Fifth Conference on Innovative Applications of Artificial Intelligence, {IAAI} 2023, Thirteenth Symposium on Educational Advances in Artificial Intelligence, {EAAI} 2023, Washington, DC, USA, February 7-14, 2023}, pages 7235--7243. {AAAI} Press, 2023.

\bibitem{DahlgaardKT17}
S{\o}ren Dahlgaard, Mathias B{\ae}k~Tejs Knudsen, and Mikkel Thorup.
\newblock Practical hash functions for similarity estimation and dimensionality reduction.
\newblock In Isabelle Guyon, Ulrike von Luxburg, Samy Bengio, Hanna~M. Wallach, Rob Fergus, S.~V.~N. Vishwanathan, and Roman Garnett, editors, {\em Advances in Neural Information Processing Systems 30: Annual Conference on Neural Information Processing Systems 2017, December 4-9, 2017, Long Beach, CA, {USA}}, pages 6615--6625, 2017.

\bibitem{datar2004locality}
Mayur Datar, Nicole Immorlica, Piotr Indyk, and Vahab~S Mirrokni.
\newblock Locality-sensitive hashing scheme based on p-stable distributions.
\newblock In {\em Proceedings of the twentieth annual symposium on Computational geometry}, pages 253--262, 2004.

\bibitem{HardtW13}
Moritz Hardt and David~P. Woodruff.
\newblock How robust are linear sketches to adaptive inputs?
\newblock In Dan Boneh, Tim Roughgarden, and Joan Feigenbaum, editors, {\em Symposium on Theory of Computing Conference, STOC'13, Palo Alto, CA, USA, June 1-4, 2013}, pages 121--130. {ACM}, 2013.

\bibitem{HassidimKMMS22}
Avinatan Hassidim, Haim Kaplan, Yishay Mansour, Yossi Matias, and Uri Stemmer.
\newblock Adversarially robust streaming algorithms via differential privacy.
\newblock {\em J. {ACM}}, 69(6):42:1--42:14, 2022.

\bibitem{IndykM98}
Piotr Indyk and Rajeev Motwani.
\newblock Approximate nearest neighbors: Towards removing the curse of dimensionality.
\newblock In Jeffrey~Scott Vitter, editor, {\em Proceedings of the Thirtieth Annual {ACM} Symposium on the Theory of Computing, Dallas, Texas, USA, May 23-26, 1998}, pages 604--613. {ACM}, 1998.

\bibitem{JMNIC21}
Omid Jafari, Preeti Maurya, Parth Nagarkar, Khandker~Mushfiqul Islam, and Chidambaram Crushev.
\newblock A survey on locality sensitive hashing algorithms and their applications.
\newblock {\em arXiv preprint arXiv:2102.08942}, 2021.

\bibitem{Kapralov15}
Michael Kapralov.
\newblock Smooth tradeoffs between insert and query complexity in nearest neighbor search.
\newblock In Tova Milo and Diego Calvanese, editors, {\em Proceedings of the 34th {ACM} Symposium on Principles of Database Systems, {PODS} 2015, Melbourne, Victoria, Australia, May 31 - June 4, 2015}, pages 329--342. {ACM}, 2015.

\bibitem{KennedyW17}
Christopher Kennedy and Rachel~A. Ward.
\newblock Fast cross-polytope locality-sensitive hashing.
\newblock In Christos~H. Papadimitriou, editor, {\em 8th Innovations in Theoretical Computer Science Conference, {ITCS} 2017, January 9-11, 2017, Berkeley, CA, {USA}}, volume~67 of {\em LIPIcs}, pages 53:1--53:16. Schloss Dagstuhl - Leibniz-Zentrum f{\"{u}}r Informatik, 2017.

\bibitem{KIW07}
Hisashi Koga, Tetsuo Ishibashi, and Toshinori Watanabe.
\newblock Fast agglomerative hierarchical clustering algorithm using locality-sensitive hashing.
\newblock {\em Knowledge and Information Systems}, 12:25--53, 2007.

\bibitem{KushilevitzOR00}
Eyal Kushilevitz, Rafail Ostrovsky, and Yuval Rabani.
\newblock Efficient search for approximate nearest neighbor in high dimensional spaces.
\newblock {\em {SIAM} J. Comput.}, 30(2):457--474, 2000.

\bibitem{lecun1998mnist}
Yann LeCun.
\newblock The mnist database of handwritten digits.
\newblock {\em http://yann. lecun. com/exdb/mnist/}, 1998.

\bibitem{Linips12}
Ping Li, Art~B. Owen, and Cun{-}Hui Zhang.
\newblock One permutation hashing.
\newblock In Peter~L. Bartlett, Fernando C.~N. Pereira, Christopher J.~C. Burges, L{\'{e}}on Bottou, and Kilian~Q. Weinberger, editors, {\em Advances in Neural Information Processing Systems 25: 26th Annual Conference on Neural Information Processing Systems 2012. Proceedings of a meeting held December 3-6, 2012, Lake Tahoe, Nevada, United States}, pages 3122--3130, 2012.

\bibitem{NelsonNW12}
Jelani Nelson, Huy~L. Nguy{\^{e}}n, and David~P. Woodruff.
\newblock On deterministic sketching and streaming for sparse recovery and norm estimation.
\newblock In Anupam Gupta, Klaus Jansen, Jos{\'{e}} D.~P. Rolim, and Rocco~A. Servedio, editors, {\em Approximation, Randomization, and Combinatorial Optimization. Algorithms and Techniques - 15th International Workshop, {APPROX} 2012, and 16th International Workshop, {RANDOM} 2012, Cambridge, MA, USA, August 15-17, 2012. Proceedings}, volume 7408 of {\em Lecture Notes in Computer Science}, pages 627--638. Springer, 2012.

\bibitem{Pagh16}
Rasmus Pagh.
\newblock Locality-sensitive hashing without false negatives.
\newblock In Robert Krauthgamer, editor, {\em Proceedings of the Twenty-Seventh Annual {ACM-SIAM} Symposium on Discrete Algorithms, {SODA} 2016, Arlington, VA, USA, January 10-12, 2016}, pages 1--9. {SIAM}, 2016.

\bibitem{Panigrahy06}
Rina Panigrahy.
\newblock Entropy based nearest neighbor search in high dimensions.
\newblock In {\em Proceedings of the Seventeenth Annual {ACM-SIAM} Symposium on Discrete Algorithms, {SODA} 2006, Miami, Florida, USA, January 22-26, 2006}, pages 1186--1195. {ACM} Press, 2006.

\bibitem{RK08}
Matti Ryynanen and Anssi Klapuri.
\newblock Query by humming of midi and audio using locality sensitive hashing.
\newblock In {\em 2008 IEEE International Conference on Acoustics, Speech and Signal Processing}, pages 2249--2252, 2008.

\bibitem{misc_mushroom_73}
{Mushroom}.
\newblock UCI Machine Learning Repository, 1987.
\newblock {DOI}: https://doi.org/10.24432/C5959T.

\bibitem{WeberSB98}
Roger Weber, Hans{-}J{\"{o}}rg Schek, and Stephen Blott.
\newblock A quantitative analysis and performance study for similarity-search methods in high-dimensional spaces.
\newblock In Ashish Gupta, Oded Shmueli, and Jennifer Widom, editors, {\em VLDB'98, Proceedings of 24rd International Conference on Very Large Data Bases, August 24-27, 1998, New York City, New York, {USA}}, pages 194--205. Morgan Kaufmann, 1998.

\bibitem{Wei22}
Alexander Wei.
\newblock Optimal las vegas approximate near neighbors in \emph{{\(\ell\)}\({}_{\mbox{p}}\)}.
\newblock {\em {ACM} Trans. Algorithms}, 18(1):7:1--7:27, 2022.

\bibitem{WoodruffZ21}
David~P. Woodruff and Samson Zhou.
\newblock Tight bounds for adversarially robust streams and sliding windows via difference estimators.
\newblock In {\em 62nd {IEEE} Annual Symposium on Foundations of Computer Science, {FOCS} 2021, Denver, CO, USA, February 7-10, 2022}, pages 1183--1196. {IEEE}, 2021.

\bibitem{wulff1982audubon}
Barry~L Wulff.
\newblock The audubon society field guide to north american mushrooms, by gary lincoff, 1982.

\bibitem{ZLZRS16}
Ying Zhang, Huchuan Lu, Lihe Zhang, Xiang Ruan, and Shun Sakai.
\newblock Video anomaly detection based on locality sensitive hashing filters.
\newblock {\em Pattern Recognition}, 59:302--311, 2016.

\end{thebibliography}

\appendix
\section{Omitted proofs and statements}
\label{sec:appendix}

This appendix contains lemmas and proofs omitted in the \cref{sec:alg_analysis}.

\subsection{Proofs from \cref{sec:alg_analysis}}
\label{app:sec41}

\begin{fact} \label{lem:exp_bounds}
For $x \in \left[0, +\infty\right)$, $1 - x \leq e^{-x}$. For $x \in \left[0, 1/5\right]$, $e^{-x - x^2} \leq 1 - x$.
\end{fact}
\begin{proof}
First we show $e^{-x} - (1 - x) \geq 0$. Taking a derivative, we get $-e^{-x} + 1$, which is non-negative on $\left[0, +\infty\right)$. Therefore, since the inequality holds at $0$, it holds on $\left[0, +\infty\right)$.

For the second inequality, consider $f(x) := e^{-x - x^2} - (1 - x)$. It is enough to show that it's first derivative $f'(x) = -(2x + 1)e^{-x - x^2} + 1$ is non-positive on $\left[0, 1/5\right]$. For this it is sufficient to show that $g(x) := (2x + 1) - e^{x + x^2} \geq 0$ on this interval. $g(x)' = 2 - (2x + 1) e^{x + x^2}$ is a decreasing function. Since $g'(0) > 0$, $g'(1/5) \approx 0.22 > 0$, $g'(x)$ is positive on the interval $\left[0, 1/5\right]$. This means that $g(x)$ is non-negative on this interval since $g(0) = 0$.
\end{proof}

\kbounds*
\begin{proof}
    Recall that $k = \lceil \frac{\ln n}{-\ln p_2} \rceil = \lceil \ln n \cdot (-\ln (1 - \frac{cr}{d}))^{-1}\rceil$. 
    By taking logarithm of the inequalities of \cref{lem:exp_bounds}, we get that for $x \leq 1/5$, $x \leq -\ln(1 - x) \leq x + x^2$. Setting $x = cr/d$, we obtain 
    \[
    \ln n \cdot \left(\frac{cr}{d}\right)^{-1} \geq \frac{\ln n}{-\ln (1 - \frac{cr}{d}))} \geq \ln n \cdot \left(\frac{cr}{d} + \left(\frac{cr}{d}\right)^2\right)^{-1}.\]
    
    The required inequalities now follow immediately, since $0 \leq cr/d \leq 1$ and $\ln n \geq 1$, so $\frac{d}{cr} \ln n \geq 1$ and $\frac{1}{1 + cr/d} \leq 1/2$.
\end{proof}

% \begin{lemma} \label{lem:support_size_help_old}
% Let $2 k^2 \le d$. Then for every $1\le x \le k$ the probability that $g$ has support size $k -x$ is at most $2^{-x}$.
% \end{lemma}

% \xxx[MM]{$k^2 \leq d$ is tight for the next lemma}

% \begin{proof}
%     Recall that a hash function $g$ is composed of hash functions $h_1,\dots, h_k$ each of which are supported on one dimension from $\{1,\dots,d\}$. Hence, the number of choices for $g$ is $d^k$. Now we would like to derive an upper bound on the number of hash functions $g$ with support size exactly $k-x$ for $1 \le x \le k$. In order to do so, we first select the $x$ indices of $k$ hash functions that are repeated, that is, a set of indices $I$ such that all $h_i, i \in \{1, \ldots, k\} \setminus I$ are pairwise distinct. There are ${k \choose x} \le k^x / x!$ choices. For the remaining (non-repeated) $k-x$ hash functions there are at most $d^{k-x}$ choices. Now for each repeated function we can specify one of the $k-x$ functions, which amounts to at most $(k-x)^x \le k^x$ choices. Therefore, the overall number of choices is at most
%     $d^{k-x} \cdot k^{2x}$ and therefore the probability to have exactly $x$ repetitions is at most $\frac{d^{k-x} k^{2x}}{d^k}= \frac{k^{2x}}{d^{x}} \le 2^{-x}$, where the last inequality follows from $2 k^2 \le d$.
% \end{proof}

\begin{fact}[Chernoff bound] \label{fact:chenoff}
Let $X_1, \ldots, X_n$ be independent Bernoulli random variables. Let $S = \sum_{i=1}^{n} X_i$. Let $\delta > 0$. Then
\[
    \Pr[S \geq (1 + \delta) \E(S)] \leq \exp \left( - \frac{\delta^2 \E(S)}{2 + \delta} \right).
\]    
\end{fact}

\begin{lemma} \label{lem:support_size_help}
Suppose that a hash function $g$ was constructed according to the procedure in \cref{sec:basics} and $n > e$. Then with probability at least $1 - 1/n^3$ the size of the support of $g$ is at least $k - 7 \log n \cdot \max\{1, \frac{k(k-1)}{2d}\}$.
\end{lemma}
\begin{proof}
    Recall that a hash function $g$ is composed of hash functions $h_1,\dots, h_k$ each of which are supported on one dimension from $\{1,\dots,d\}$. The dimension on which each of the function $h_i$ is supported is picked uniformly and independently from $\{1, \ldots, d\}$. We will call hash function $h_i$ duplicated if for some $j < i$, $h_i = h_j$. Note that the first instance of each hash function in the range according to this definition is not duplicated. Then, because $k = |\text{support of } g| + \# \text{duplicated hash functions } h$, it is enough to bound the number of duplicated hash functions to conclude the statement of the lemma, which we will denote by $S$.

    Consider the following equivalent process of constructing $g$. We will pick each of the $h_i$ in a sequence of increasing $i$. Denote by $t_i$ the size of the union of supports of $h_1, \ldots, h_{i-1}$. For each $i$ we will first reorder the range $\{1, \ldots, d\}$ so that the first $t_i$ entries would be the dimension indices already in the support of $h_1, \ldots, h_{i - 1}$, and the rest would be everything else. We will then pick a number $a_i \in [d]$ uniformly at random, and this number would indicate the position of the dimension index in the new range on which $h_i$ will be supported on.

    It is now easy to see that the hash function $h_i$ is duplicated iff $a_i \leq t_i$.

    Since the support of first $i - 1$ hash functions is at most of size $i - 1$, it always holds that $t_i \leq i - 1$. Let $\sigma_i$ be an indicator of the event $a_i \leq t_i$, and $\xi_i$ be an indicator of the event $a_i \leq i - 1$.
    Then $\sum_{i = 1}^{k} \sigma_i = S$ and for each $i$, $\sigma_i \leq \xi_i$.

    Hence we will upper bound $S$ by $S' := \sum_{i = 1}^{k} \xi_i$, and upper bound $S'$ using the Chernoff bound.
    Indeed $\E(S') = \sum_{i = 1}^{k} \E(\xi_i) = \sum_{i = 1}^{k} \frac{i - 1}{d} = \frac{k(k-1)}{2d}$.
    Note that $\delta^2 / (2 + \delta) \geq \delta / 2$ when $\delta > 2$, which is the case for the values we use below.
    
    Now, we first consider the case when $\frac{k(k - 1)}{2d} \geq 1$.
    
    By applying \cref{fact:chenoff} to $\xi_i$ and setting $\delta = 6 \ln n$, we get 
    \begin{equation*}
    \begin{split}
        \Pr\left[S' \geq 7 \ln n \frac{k(k-1)}{2d}\right] &\leq \Pr\left[S' \geq (1 + 6 \ln n) \frac{k(k-1)}{2d}\right] \\
        &\leq \exp \left(-\frac{36 \ln^2 n}{2 + 6 \ln n} \cdot \frac{k(k-1)}{2d}\right)\\
        &\leq \exp \left(- 3 \ln n \frac{k(k-1)}{2d}\right)\\
        &\leq 1/n^3.
    \end{split}        
    \end{equation*}

    When $\frac{k(k - 1)}{2d} < 1$, set $\delta = 6 \ln n \frac{2d}{k(k-1)}$.
    \begin{equation*}
    \begin{split}
        \Pr\left[S' \geq 7 \ln n\right] &\leq \Pr\left[S' \geq (1 + 6 \ln n \frac{2d}{k(k-1)}) \frac{k(k-1)}{2d}\right] \\
        &\leq \exp \left(- 3 \ln n \frac{2d}{k(k-1)} \cdot \frac{k(k-1)}{2d}\right)\\
        &\leq \exp \left(- 3 \ln n \right) \\
        &\leq 1/n^3.
    \end{split}        
    \end{equation*}

    Therefore, 
    \begin{multline*}
    \Pr\left[|\text{support of } g| < k - 7 \log n \cdot \max\left\{1, \frac{k(k-1)}{2d}\right\}\right]  \\
    =\Pr\left[S \geq k - 7 \log n \cdot \max\left\{1, \frac{k(k-1)}{2d}\right\}\right]  \\
    \leq\Pr\left[S' \geq k - 7 \log n \cdot \max\left\{1, \frac{k(k-1)}{2d}\right\}\right]\\
    \leq 1/n^3,
    \end{multline*}
    which concludes the proof.
\end{proof}

\supplb*
\begin{proof} 
By \cref{lem:support_size_help}, the probability that the support of any fixed hash function in $G$ is at least $k - 7 \log n \cdot \max\{1, \frac{k(k-1)}{2d}\}$ is at least $1 - 1/n^3$. The number of hash functions in $G$ is $\lceil \lambda n^\rho\rceil \leq n^2$ and so
by the union bound, the probability that there exists a hash function $g\in G$ with support less than $k - 7 \log n \cdot \max\{1, \frac{k^2}{2d}\}$ is at most
$ \lceil \lambda n^{\rho} \rceil \cdot \frac{1}{n^3} \le 1/n$.
\end{proof}

\nearcoll*

\begin{proof}
    We fix an arbitrary set $G$ of $L$ hash functions that satisfy the condition of the lemma. 
    By linearity of expectation, 
    $$
    \E [|\coll(q,z)|] = \sum_{g\in G} \Pr[g(q) = q(z)],
    $$
    where the randomness is over the choice of $q$ as in the lemma.
    
    Let $x = 7 \ln n \cdot \max\{1, \frac{k^2}{2d}\}$
    
    In the following we consider an arbitrary $g\in G$
    and derive an upper bound for the collision probability
    of $q$ and $z$.

    \begin{equation}\label{eq:0wireg90hweg9}
    \Pr_{q\sim UNIF(\mathbb{S}(z, r - t))}[g(q) = g(z)] \leq \left(1 - \frac{k - x}{d}\right)^{r- t},
    \end{equation}
    where $\mathbb{S}(z, r - t) = \{ p \in \{0, 1\}^d: \dist(z, p) = r-t \}$.

    To see it, consider the following way to sample $q$: first sample $r - t$ bits uniformly at random with replacement from $[d]$ to form a set $B$, $|B| \leq r- t$, and then flip the bits from $B$ in $z$ to obtain a point $q'$. After that choose the rest of $r - t - \dist(z, q')$ bits without replacement uniformly at random among the bits $[d] \setminus B$ and flip them. Because the bits flipped in the second step are different from the bits flipped in the first step, if $g(q') \neq g(z)$ then $g(q) \neq g(z)$. Therefore, $\Pr[g(q') = g(z)] \geq \Pr[g(q) = g(z)]$.

    On the other hand, $g(q') = g(z)$ if and only if all of the bits sampled into $B$ do not lie in the support of $g$.
    Because the support of $g$ is at least $k - x$, each sampled bit doesn't lie in it with probability at most $1 - \frac{k - x}{d}$. Since the bits are chosen independently, we have $\Pr[g(q') = g(z)] \leq (1 - \frac{k-x}{d})^{r - t}$. This establishes~\eqref{eq:0wireg90hweg9}.

    We are now in a position to upper bound the expected size of $\coll(q, z)$, i.e. the number of hashings that $q$ and $z$ collide under, obtaining the result of the lemma.  Using~\eqref{eq:0wireg90hweg9} we get    
    \begin{multline}
    \E_{q}\left[ |\coll(q, z)| \right] = \E_{q}\left[\sum_{g \in G} I(g(q) = g(z))\right] =\\
    \sum_{g \in G} \Pr[g(q) = g(z)] \leq L \cdot \left(1 - \frac{k - x}{d}\right)^{r -t}.
    \end{multline}

    Recall that by our parameter setting the number of hashings $L$  satisfies $L = \lceil \lambda \ell \rceil \leq (\lambda + 1) \ell$, where 
    \begin{equation}\label{eq:whg9whg90hwe9g0h}
        \ell = n^{\rho} = \left(1 - \frac{r}{d}\right)^{-\log_{1/p_2} n} \leq \left(1 - \frac{r}{d}\right)^{-k}.
    \end{equation}

    Hence, by~\eqref{eq:whg9whg90hwe9g0h}, we have the following upper bound on the expectation.
    \begin{equation} \label{eq:3g9h3490fhohfishdifh}
    \E_{q}\left[ |\coll(q, z)| \right] \leq (\lambda + 1) \left(1 - \frac{r}{d}\right)^{-k} \left(1 - \frac{k - x}{d}\right)^{r -t}
    \end{equation}

    The only step left is to upper bound the right hand side of~\eqref{eq:3g9h3490fhohfishdifh}.
    Therefore, we need to upper bound $\left(1 - \frac{r}{d}\right)^{-k} \left(1 - \frac{k - x}{d}\right)^{r -t}$. We have by \cref{lem:exp_bounds}
    \begin{align*}
        (1 - \frac{k - x}{d})^{r -t} & \leq \exp(-\frac{k - x}{d})^{r-t}\\
        & = \exp(-\frac{kr}{d} + \frac{x(r-t)}{d} + \frac{tk}{d}) \\
        (1 - \frac{r}{d})^{k} & \geq \exp(-\frac{r}{d} + \frac{r^2}{d^2})^k \\
        & = \exp(-\frac{kr}{d} - \frac{kr^2}{d^2})
    \end{align*}

    The largest term $\frac{kr}{d}$ cancels out in the fraction, so we are only left with a sum of smaller terms.
    
    \begin{equation}
         \label{eq:w0hg90wh9ghef}
        \frac{(1 - \frac{k - x}{d})^{r -t}}{(1 - \frac{r}{d})^{k}}
        \leq
        \exp\left( \frac{tk}{d} + \frac{kr^2}{d^2} + \frac{x(r - t)}{d} \right).
    \end{equation}

    Now we will establish upper bounds on each of the summands in the right-hand side of~\eqref{eq:w0hg90wh9ghef}. First, by \cref{cor:k_bounds},
    \[
        k \leq 2 \frac{d}{cr} \ln n.
    \]
    This implies:
    \begin{itemize}
        \item $k \leq d/(2t)$, since $8e^2 (\lambda + 1) \ln = 4 t \ln < cr$,
        so $2 \frac{d}{cr} \ln n < d/(2t)$.
        \item $k \leq d^2/r^2$, since $r \leq d/(14 \ln n) \leq \frac{dc}{2 \ln n}$, so
        $2 \frac{1}{c} \ln n \leq d/r$ and $2 \frac{d}{cr} \ln n \leq d^2/r^2$.
        \item $k^2 \leq \frac{d^2}{7 r \ln n}$, since $ r \geq 28 \ln^3 n / c^2$,
        so $\frac{1}{7 \ln n } \geq 4 \frac{1}{c^2r} \ln^2 n$ and $4 \frac{d^2}{c^2r^2} \ln^2 n \leq \frac{d^2}{7 r \ln n}$
    \end{itemize}
    
    The last inequality implies that $7 \ln n r \cdot \frac{k^2}{2d^2} \leq 1/2$, and
    $r \leq \frac{d}{14 \ln n}$ implies $7 \ln n r / d \leq 1/2$, so 
    \[
    \frac{x(r-t)}{d} \leq xr/d \leq 1/2.
    \]

    On the other hand, the first two bounds on $k$ imply $tk/d < 1/2$ and $kr^2/d^2 < 1$, which finally allows us to bound rhs of~\eqref{eq:w0hg90wh9ghef}:
    $$
    \exp\left( \frac{tk}{d} + \frac{kr^2}{d^2}  + \frac{x(r - t)}{d} \right) \leq \exp\left( 0.5 + 1  + 0.5 \right) = e^2.
    $$

    Putting the above together with~\eqref{eq:3g9h3490fhohfishdifh} we therefore conclude that
    $$
    \E_{q}[\coll(q, z)] \leq e^2 (\lambda + 1)
    $$
    as required.
\end{proof}

\subsection{Proofs from \cref{sec:simple}} \label{app:simple_analysis}

We split \cref{thm:combined} into two separate \cref{thm:simple,thm:fast} for each of \cref{alg:simple,alg:fast}. In this subsection we formally prove \cref{thm:simple}.

\begin{restatable}[Simplest algorithm]{theorem}{simple} \label{thm:simple}
Let $n>e$, $28 \ln^3 n \leq r \leq d/(28 \ln n)$, $\lambda \leq \frac{1}{8e^2} \min\{\frac{r}{\ln n}, n^{1/8}\} - 1$, and $1 + \frac{80 + 16 \ln(\lambda + 1)}{\ln n} \leq c \leq \ln n$.
%Let $n > e$ and
%\begin{itemize}
%    \item $28 \ln^3 n \leq r \leq d/(28 \ln n)$,
%    \item $\lambda \leq \frac{1}{8e^2} \min\{\frac{r}{\ln n}, n^{1/8}\} - 1$,
%    \item $1 + \frac{80 + 16 \ln(\lambda + 1)}{\ln n} \leq c \leq \ln n$.
%\end{itemize}
Let $t = 2e^2 (\lambda + 1)$ be the parameter from \cref{alg:simple}. Suppose that the point $z$  passed to \cref{alg:simple} is located at a distance of at least $2cr$ from any other point in $P$. With probability at least $1/4 - 1/n$ the algorithm finds a point $q$ at a distance of at most $r$ from $z$ such that querying LSH with $q$ returns no point. It uses at most $O(cr \cdot \lambda)$ queries to the LSH data structure.
\end{restatable}

\begin{proof}   
    First we show that under this parameter setting the requirements of \cref{lem:far_coll,lem:near_coll} are satisfied. For this we need to show that $r \leq d / (5c)$ and that the support of each hash function $g \in G$ is at least $k/2$.

    The inequality $r \leq d / (5c)$ immediately follows from $c \leq \ln n$ and $r \leq d/(28 \ln)$. 
 
    By \cref{lem:supp_lb}, each hash function $g \in G$ has support at least $k - 7 \ln n \cdot \max\{1, \frac{k^2}{2d}\}$ with probability at least $1/n$.
    By \cref{cor:k_bounds} and the parameter setting, $r \leq d/(28 \ln n) \leq d/(28 c)$, hence $14 \ln n \leq \frac{d}{2cr} \ln n \leq k$. On the other hand, $r 
    \geq 28 \ln^3 n \geq 14 \ln n / c$, so $k \leq \frac{2d}{cr} \ln n \leq \frac{d}{7 \ln}$ and $\frac{k^2}{2d} \leq \frac{k}{14 \ln n}$. This means that $k/2 \geq 7 \ln n \cdot \max\{1, \frac{k^2}{2d}\}$, so each $g \in G$ has support at least $k/2$ with probality at least $1 - 1/n$.
    
    To conclude the theorem statement, we show in this proof that \cref{alg:simple} returns a point $q$ within distance $r$ of $z$ such that $\coll(q, z) = \emptyset$.
    
    Let $q$ be the point sampled in line~\ref{line:sample_point}. By \cref{lem:near_coll} and Markov's inequality, $|\coll(q, z)| \leq 2e^2 (\lambda + 1)$ with probability $1/2$. Conditioned on those inequalities holding, we will show that with a probability at least $1 - \frac{1}{8e^2 (\lambda + 1)}$ the algorithm finds a bit $i$ during each execution of the outer loop at line~\ref{line:outer_loop} such that, after flipping this bit in $q$, $|\coll(q, z)|$ decreases by at least $1$. This would mean that after $2e^2 (\lambda + 1)$ iterations $\coll(q, z) = \emptyset$. This happens with probability at least $1 - (1/2 + 2e^2 (\lambda + 1) \frac{1}{8 e^2 (\lambda + 1)}) = 1/4 $ by a union bound.

    By taking the union bound with the probability that each $g \in G$ has support at least $k - 4 \log n$, we obtain the final success probability of at least $1/4 - 1/n$.

    The inner loop at line~\ref{line:slow_inner_start} by \cref{lem:simple_inner_loop}, with probability at least \( 1- 1/(32 e^2(\lambda + 1)) \), finds two points $q'$ and $q''$ such that $\coll(q'', z) \neq \emptyset$, $\coll(q', z) = \emptyset$ and $q',q''$ differ in exactly one bit.

    Assuming such a pair of points is found, let $i$ be the position of the bit that is different between them. Since $\coll(q'', z) \neq \emptyset$, there is a function $g \in \coll(q, z)$ such that $g(q'') = g(z)$. On the other hand, because $\coll(q',z) = \emptyset$, $g(q') \neq g(z)$, therefore bit $i$ must lie in the support of $g$. Hence after we flip this bit in $q$, $g(q) \neq g(z)$ and $|\coll(q, z)|$ decreases.

    Note that because we are doing the inner loop at most $t$ times, $\dist(q, z) \leq r$, so the prerequisites for \cref{lem:simple_inner_loop} always hold.

    The outer loop is executed at most $2e^2 (\lambda + 1)$ times, and the inner loop makes at most $cr - (r - t)$ iterations, as each iteration moves $q'$ further away from $z$. Therefore, the total number of queries is bounded by $O((cr - r + t) \lambda )$.
\end{proof}

\subsection{Proofs for \cref{alg:fast}}
\label{sec:fast_analysis}

Here we present \cref{alg:fast} and its analysis. The key difference compared to \cref{alg:simple} is the inner loop, which we analyze in a separate lemma.

\begin{algorithm}[h]
\caption{Fast adaptive adversarial walk from $z$. } \label{alg:fast}
\begin{algorithmic}[1]
\REQUIRE{Origin point $z$, parameters $r$, $cr$, $\lambda$, access to an \lsh{} structure}
\STATE $t \gets 2 e^2 (\lambda + 1) $
\STATE Sample $q$ at distance $r - t$ from $z$ \label{line:fast_sample_point}
\WHILE{$\lsh(q) = z$} \label{line:fast_outer_loop}
    \IF{$\dist(q, z) \geq r$} \label{line:radius_cutoff}
    \RETURN $\bot$
    \ENDIF
    \STATE $q^{left}, q^{right} \gets q$ \label{line:fast_inner_start}
    \STATE Choose $cr - \dist(q^{right}, z)$ random bits $i$ in $q^{right}$ such that $q^{left}_i \neq z_i$ and flip them
    \IF{$\lsh(q^{right}) = z$}
        \RETURN $\bot$ \label{line:fast_return_bot}
    \ENDIF
    \WHILE{$\dist(q^{left}, q^{right}) > 1$} \label{line:fast_inner_loop}
        \STATE $q^{mid} \gets q^{left}$
        \STATE Flip any $\frac{1}{2} \dist(q^{left}, q^{right})$ bits $i$ in $q^{mid}$ such that $q^{left}_i \neq q^{right}_i$
        \IF{$\lsh(q^{mid}) = z$}
            \STATE $q^{left} \gets q^{mid}$
        \ELSE
            \STATE $q^{right} \gets q^{mid}$
        \ENDIF
    \ENDWHILE
    \STATE $j \gets$ the bit where $q^{left}$ differs from $q^{right}$ \label{line:fast_inner_end}
    \STATE Flip bit $j$ in $q$
\ENDWHILE
\RETURN $q$
\end{algorithmic}
\end{algorithm}

\begin{lemma}[Inner loop of the fast adversarial walk] \label{lem:fast_inner_loop}
    Let $n > e$, $cr/d \leq 1/5$, $c \geq 1 + \frac{80 + 16 \ln(\lambda + 1)}{\ln n}$, $8e^2(\lambda + 1) \leq n^{1/8}$. Let $G$ be such that no hash functions $g \in G$ has support smaller than $k - 4 \log n$. Suppose that the point $z$  passed to \cref{alg:fast} is located at distance at least $2cr$ from any other point in $P$. Let $q$ sampled in line~\ref{line:fast_sample_point} be such that $|\coll(q, z)| \leq 2 e^2 (\lambda + 1)$.
    For a fixed iteration of the outer loop in line~\ref{line:fast_outer_loop} of \cref{alg:fast}, with probability at least \(1 - 1/(32e^2 (\lambda + 1))\) 
    the inner loop between the lines~\ref{line:fast_inner_start} and \ref{line:fast_inner_end} finds two points $q^{left}$ and $q^{right}$ such that:
    \begin{itemize}
        \item $\lsh(q^{right}) \neq z$,
        \item $\lsh(q^{left}) = z$,
        \item $q^{left}$ and $q^{right}$ differ in one bit.
    \end{itemize}
    The inner loops takes time at most $O(\log(cr - r + t))$.
\end{lemma}
\begin{proof}
    Fix an iteration of the outer loop. The idea of the inner loop is to sample two points $q^{left}$ and $q^{right}$ such that there is a path between them that goes strictly away from $z$ and such that $\lsh(q^{left}) = z$ and $\lsh(q^{right}) \neq z$.
    Then we move those points towards each other by binary search while preserving those properties.

    Initial point $q^{right}$ is sampled through the same process as in \cref{lem:far_coll}, so with probability at least \(1 - 1/(32e^2 (\lambda + 1)\) the condition $\lsh(q^{right}) \neq z$ holds and the return in line~\ref{line:fast_return_bot} is not reached. The condition $\lsh(q^{left}) = z$ holds at the start because $\lsh(q) = z$ by the while clause of the outer loop.

    During the execution both guarantees continue to hold through the way we reassign variables in the binary search. At the end, the distance between the two points becomes $1$, which means that they are differ in exactly one bit. This concludes the correctness proof.
    
    The running time is $O(\log(cr - (r - t))$ because the distance between $q^{left}$ and $q^{right}$ is halved after each iteration.
\end{proof}

The proof of \cref{alg:fast} is almost the same as of \cref{alg:simple}, so we only outline the differences.

\begin{theorem}[Fast algorithm] \label{thm:fast}
Let $n>e$, $28 \ln^3 n \leq r \leq d/(28 \ln n)$, $\lambda \leq \frac{1}{8e^2} \min\{\frac{r}{\ln n}, n^{1/8}\} - 1$, and $1 + \frac{80 + 16 \ln(\lambda + 1)}{\ln n} \leq c \leq \ln n$.
%
%Let $n > e$ and
%\begin{itemize}
%    \item $28 \ln^3 n \leq r \leq d/(28 \ln n)$,
%    \item $\lambda \leq \frac{1}{8e^2} \min\{\frac{r}{\ln n}, n^{1/8}\} %- 1$,
%    \item $1 + \frac{80 + 16 \ln(\lambda + 1)}{\ln n} \leq c \leq \ln n$.
%\end{itemize}
Let $t = 2e^2 (\lambda + 1)$ be the parameter from \cref{alg:simple} (\cref{alg:fast}). Suppose that the point $z$  passed to \cref{alg:simple} 
(\cref{alg:fast})
 is located at a distance of at least $2cr$ from any other point in $P$. With probability at least $1/4 - 1/n$ the algorithm finds a point $q$ at a distance of at most $r$ from $z$ such that querying LSH with $q$ returns no point.
It uses at most $O(cr\cdot \lambda)$ ($O(\log (cr) \cdot \lambda)$ for 
\cref{alg:fast}) queries to the LSH data structure.
\end{theorem}

\begin{proof}
    \Cref{alg:fast,alg:simple} only differ in the implementation of the inner loop, between the lines \ref{line:slow_inner_start} and \ref{line:slow_inner_end} and the lines~\ref{line:fast_inner_start} and \ref{line:fast_inner_end} respectively. Because of this, the proof is exactly the same as in \cref{alg:simple}, however \cref{lem:fast_inner_loop} is used instead of \cref{lem:simple_inner_loop}.

    The final runtime is $O(\log (cr - r + t) \cdot \lambda)$ since the inner loop takes time $O(\log (cr - r + t))$.
\end{proof}

\subsection{Existence of an isolated point}

Not all data sets contain a sufficiently isolated point as required by \cref{thm:simple,thm:fast}, i.e.  a point that is at a distance $2cr$ from any other point in the set. However,  when a point set is sampled randomly, this assumption holds when $2cr < d/4$ for all of the points with high probability.
% This particular instance is interesting because it is an example of an instance where the LSH construction is tight, that is the number of hash functions cannot be reduced further without worsening the guarantees. \xxx[MM]{Add citation.}

\begin{lemma} \label{lem:rand_point_sep}
Let $P$ be a set of $n$ random points sampled uniformly independently at random from $\{0, 1\}^d$. If $d \geq 48 \ln n$, with probability at least $1 - 1/n$ the distance between any two of them is at least $d/4$.
\end{lemma}
\begin{proof}
    Fix an arbitrary point $z$. We will first show that the probability that $\dist(z, x) < d/4$ is at most $1/n^3$ for a uniformly random point $x$.

    Notice that $\dist(z, x) = \sum_{i = 1}^d y_i$, where $y_i = |z_i - x_i|$, and all $y_i$ are independent with respect to each other and are equal to $1$ with probability $1/2$ and $0$ otherwise. Hence $\E[\sum_{i=1}^{d} y_i] = d/2$. Therefore, by the Chernoff bound,
    \[
        \Pr[\dist(z, x) \leq 0.5 \cdot 0.5 d] \leq e^{-\frac{0.25 \cdot 0.5 d}{2}} \leq e^{-3 \ln n} = 1/n^3.
    \]

    Now, by a union bound across all pairs of points, the probability that the minimum of pairwise distances between any two of them is at most $d/4$ is at most $1/n$.
\end{proof}

\section{Discussion of impact} \label{sec:impact}

This paper presents work whose goal is to advance the field of Machine Learning.

The main contribution of the paper is to get an improved understanding for which choices of parameters LSH data structure are unsafe with respect to adversarial attacks. We show that, in principle, attacks on standard LSH constructions are possible, if the parameters are not chosen carefully. As a consequence  false-negative-free LSH data structures should be given preference in sensitive applications. Our main contribution is the theoretical analysis of the attack and our experiments are focussing on understanding the parameter setting for which attack may work. Therefore we believe that there is very limited risk to abuse our results and we do think that the benefit of understanding such potential attacks clearly outweights this limited risk. 
\section{Additional experiments and details} \label{app:exp}

\paragraph{Resources used.} The experiments were run on a server in an internal cluster, where a typical worker has two AMD EPYC 7302 16-Core processors, with up to 500 GiB of RAM. Experiments where run in parallel on multiple cores. In total, running all of the experiments, including those in the appendix, takes at most 100000 CPU hours, and uses 21 GiB of disk memory.

\paragraph{Data sets used in experiment 2.}
   \textit{MNIST \cite{lecun1998mnist}.} This dataset consists of grayscale pictures of handwritten digits of size $28 \times 28$. We first flatten the images into vectors of size $784$, and then we project those vectors into Hamming space by setting each entry to $0$ if it corresponded to a white pixel, i.e. if the grayscale color was $0$, and to $1$ otherwise. Note that this transformation preserves the dataset fairly well, as the background pixels in each image are perfectly white.
    
     \textit{Anonymous Microsoft Web Data (MSWeb) \cite{misc_anonymous_microsoft_web_data_4}.} Entries in this dataset are anonymous web user IDs and paths on the Microsoft website that they have visited. There are in total $294$ different paths. We represent each user by a point in $294$ dimensional Hamming space, where a point entry is equal to $1$ if they have visited the corresponding path.
     
    \textit{Mushroom \cite{misc_mushroom_73,wulff1982audubon}.} This dataset consists of descriptions of samples of mushrooms of $23$ different species. Each mushroom has $22$ categorical features describing it. We one-hot encode each feature to map the descriptions into Hamming space, with the final number of dimensions being $116$.

\paragraph{Experiment 3: Influence of $\lambda$ on the algorithm efficiency.}

In this experiment the goal was to measure the impact of $\lambda$ on the probability of success of the algorithm.
The setup is the same as in Experiment 2. All runs are done on the Random dataset. The results are on \cref{fig:delta_plot}.

Similar to Experiment 2, we compare the different values of $\lambda$ by varying the distance at which we want to find the negative query. As is evident from the plots, the bigger the $\lambda$, the harder it is for the algorithm to succeed.

Note that the value of $\lambda$ where algorithm no longer succeeds for distance $r$ are rather modest. It may be that it is enough to just increase the value of $\lambda$ to defend against the type of attack that is described in this paper.

\paragraph{Experiment 4: Number of queries as a function of starting distance.}

In this experiment we experimentally confirm that the number of queries that an algorithm needs is linearly dependent on $t$.
We use the same datasets and experimental setting as in Experiment 2. The results are on \cref{fig:all_data_queries_target_dist}.

Combined with results of Experiment 2, we can see that lowering the starting distance provides a trade-off between the running time and the probability of success.

\paragraph{Experiment 5: Efficiency gain over random sampling.}
Finally, we investigated the efficiency gain that our adaptive algorithm achieved over random point sampling. 
The experimental setup is the same as in Experiment 2. Here we use the Random dataset.
For both approaches, we continued rerunning them until a false negative query was found.
The results are presented in \cref{fig:adaptive-random}.

As you can see, when the $\lambda$ and, accordingly, the number of hash functions $L$ are very small, random sampling is more efficient. However, for even modest values of $\lambda$, we already get a significant exponential speed-up.

\paragraph{Experiment 6: Comparison of defense mechanisms.}
In this experiment we compare two different defence mechanisms based on ideas from differential privacy \cite{BeimelKMNSS22} and distance estimation \cite{CherapanamjeriN20}. 

Theorem 1.3. of \cite{CherapanamjeriN20} presents a general robustification mechanism which, when applied to then LSH of \cite{IndykM98}, says that it can be made robust with blow-ups of $O(d \log 1/\delta)$ for memory and $O(\log 1/\delta)$ for query time, where $\delta$ is the failure probability per adaptive query. This is achieved by creating $O(d \log 1/\delta)$ copies of plain LSH. Querying is done by selecting $O(\log 1/\delta)$ of those LSH data structures, relaying the query to them and then choosing any of the returned points as the answer, or $\bot$ if all of the data structure failed to find a point. 

As a side note, the original ADE construction of \cite{CherapanamjeriN20} would also allow to solve the ANN problem. We don't study it in this experiment for two reasons: it would require time per query $\Omega(n)$, which is too slow in our parameter setting, and it is based upon linear sketches, against which our adaptive adversary is not effective.

The generic differentially private conversion of \cite{BeimelKMNSS22} is similar. For a generic estimation data structure, to be able to process $T$ queries against an adaptive adversary, they create $\wt{O}(\sqrt{T})$ copies of the data structure. For each query, they select $\wt{O}(1)$ copies of the data structure and aggregate their output using a differentially private median.
Although their construction is not generally applicable to ANN, as it is a search problem, it would nevertheless still be effective against our adversary, as for the latter it is not important which exact point was returned, but only whether anything was returned at all. 

As both of the data structures can provably defend against an adaptive adversary and because both of them require a lot of additional memory, in this experiment we consider their performance against our adversary when the available memory is limited. Similarly, because rigorous implementation of a differentially private median would be computationally expensive, we instead opt to use a slighlty simplified version adapted to our setting. Namely, after sampling and querying the data structures, we count the number of them that returned a point and those that did not. We add to both values two-sided geometric noise, whose probability mass function is$\Pr[z] = \frac{1 - \alpha}{1 + \alpha} \alpha^{|z|}$, with parameter $\alpha = e^{-1/4}$. If after this perturbation the number of data structures that returned $\bot$ is bigger than the ones that returned a point, we return $\bot$. Otherwise, one of the outputs of the successful data structures is returned.

The results are depicted on \cref{fig:nc-defense,fig:dp-defense}. On both figures, plots annotated with both $\lambda$ and ``q.s.'' correspond to one of the two robustified variation accordingly. $\lambda$ represents the number of copies used, and ``q.s.'' stands for ``query samples'', and represents the number of data structures randomly sampled to answer each query. Each used copy is build with the same parameters as in experiment $2$, and $\lambda = 1$. The plots annotated with only $\lambda$ are the plain LSH constructions with different values of $\lambda$. The idea is that the plots with the same value of $\lambda$ have the same number of initialized hash functions. The used dataset is Random.

Both figures contain 4 subplots. The first subplot represents the probability that our adaptive adversary reports that it has found a false negative query at a desired target distance. Note that here, unlike in our previous experiments, the output of each query is non-deterministic even after the randomness of the hash functions is fixed. Hence we also test how ``good'' of a false negative the reported point is. For each found point, we repeatedly query it to the data structure $100$ times, and depict on the other three subplots the probability of the points discovered by the adversary to still be a negative query for at least $90\%$, $50\%$ and $10\%$ of the repeated queries.

As you can see from the \cref{fig:nc-defense}, the defense based on \cite{CherapanamjeriN20} is performing better than querying the plain LSH with the same and even often greater amount of hash functions. Similarly, the probability that the discovered false negative remain false negatives in at least $50\%$ or $90\%$ of cases is fairly low.

On the other hand, the results for the algorithm based on differential privacy on \cref{fig:dp-defense} show that the probability of reporting a false negative is considerably higher than for a plain LSH even for low values of target distance. This is due to the aggregation mechanism that we use, as it has a chance of reporting nothing even if all sampled data structures returned a positive answer. However, this also means that the adversary has less chances to learn the underlying data structure, which we can see in the fact that the probability of that a repeated query would be a negative one in at least $10\%$ of cases being below $0.5$.

Based on this, it appears that unless in your application you are prepared to tolerate a significantly higher chance of failure for each query, it is better to use the construction based on \cite{CherapanamjeriN20}. Otherwise, a differentially private construction would leak less information to the adversary. Note that both constructions are still computationally faster than directly querying each data structure.

% \begin{figure*}
%     \centering
%     \includegraphics[width=\textwidth]{figs/radius-vs-lambda-heat.png}
%     \caption{Success probability dependence on $\lambda$ and the near radius. Default parameter setting. The dashed line is where $r = 4 \lambda$ and is an approximate lower bound on $r$ in terms of $\lambda$ for the algorithm to succeed with high probability.}
%     \label{fig:radius-lambda}
% \end{figure*}

% \begin{figure*}
%     \centering
%     \includegraphics[width=\textwidth]{figs/radius-vs-dimension-heat.png}
%     \caption{Success probability dependence on the number of dimensions $d$ and the near radius $r$. Default parameter setting. The line $r = 0.15 d$ represents a parameter setting for the near radius.}
%     \label{fig:radius-dimension}
% \end{figure*}

\begin{figure*}[h]
  \centering
  \includegraphics[width=\textwidth]{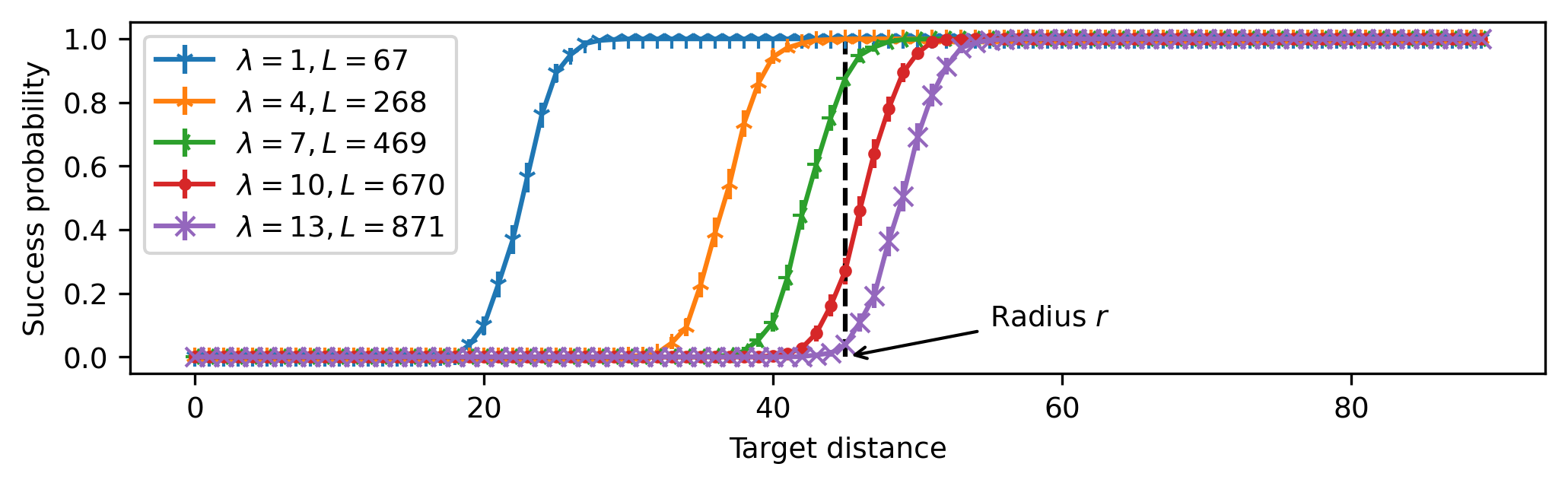}
  \caption{ Desired distance of the query not hashing with the origin point from the origin point $z$ for different values of $\lambda$. The experiments are conducted on the Zero dataset for the default parameter setting. The distance equal to $r$ is denoted by a dashed vertical line.}
  \label{fig:delta_plot}
\end{figure*}

\begin{figure*}[h]
  \centering
  \includegraphics[width=\textwidth]{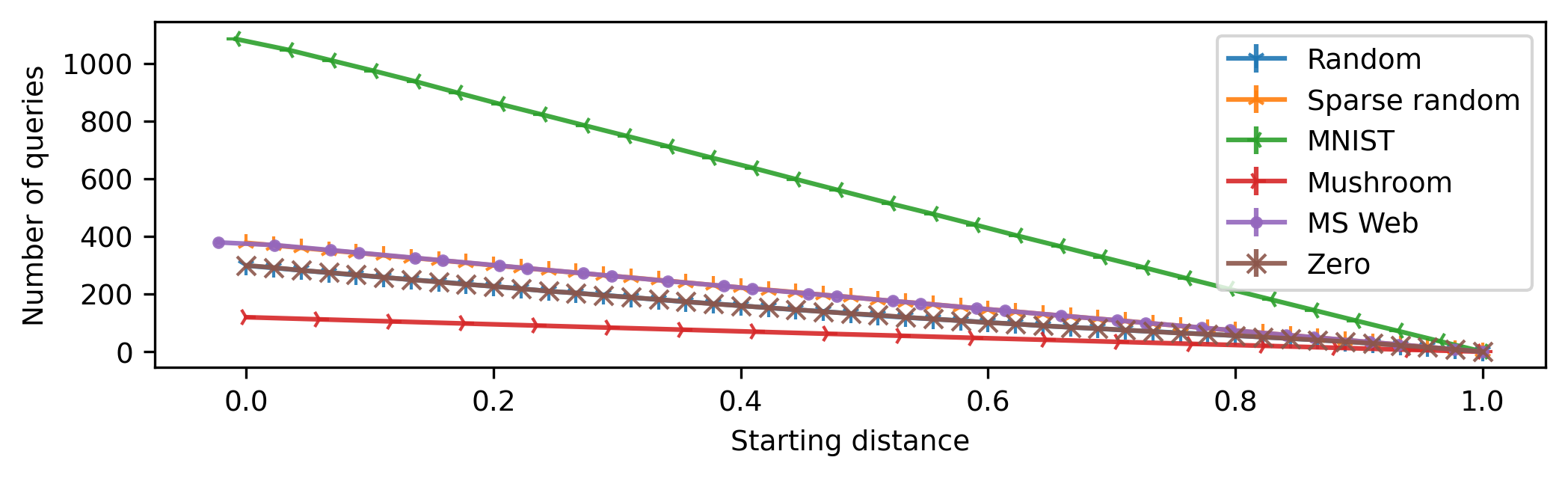}
  \caption{The number of queries as a function of the distance of $q$ from origin point $z$ on all tested datasets. }
  \label{fig:all_data_queries_target_dist}
\end{figure*}

\begin{figure*}[h]
  \centering
  \includegraphics[width=\textwidth]{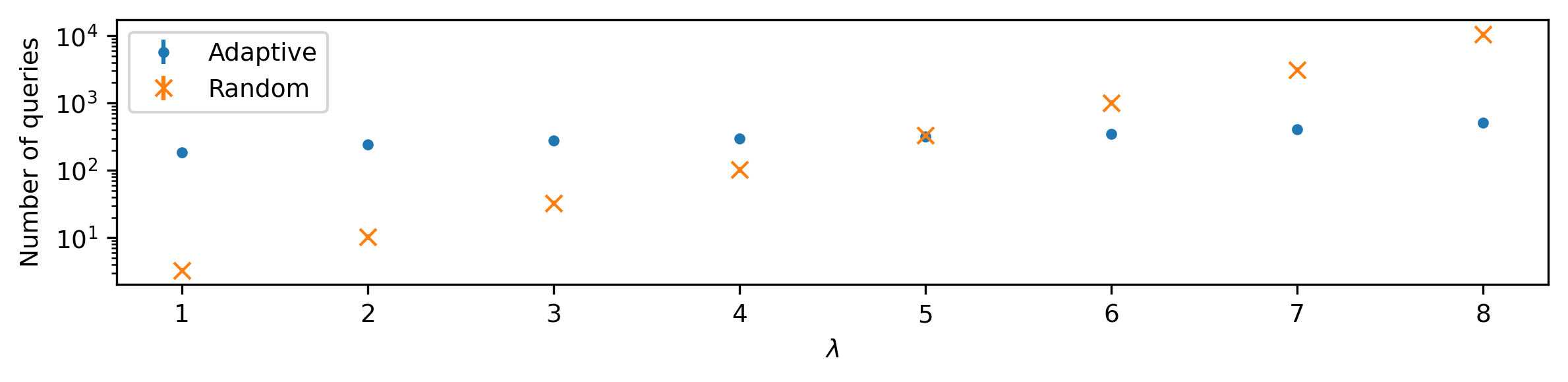}
  \caption{Mean number of queries necessary to find a false negative depending on parameter $\lambda$ required by our adaptive adversary and random sampling of query points.}
  \label{fig:adaptive-random}
\end{figure*}

\begin{figure*}[h]
  \centering
  \includegraphics[width=\textwidth]{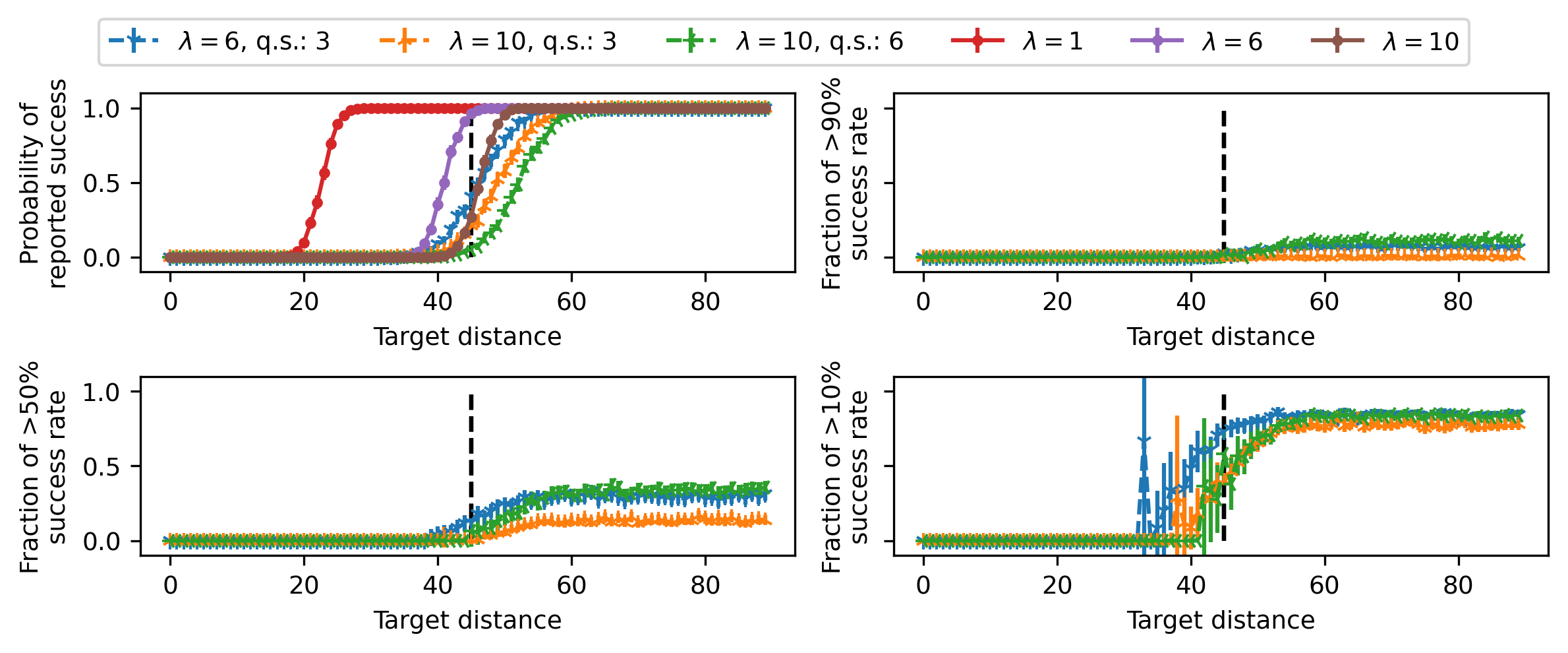}
  \caption{Performance of robustification based on \cite{CherapanamjeriN20}. First subplot depicts probability that our adaptive adversary finds a false negative query versus the desired distance to that query from origin. The rest of subplots depict probabilities that, if such a query was found, if it was queried again it would still be a false negative query with probability at least $90\%$, $50\%$ or $10\%$. In the legend, plots annotated with both $\lambda$ and ``q.s'' correspond to the robust algorithm, with ``q.s.'' standing for ``query samples'', the number of data structures sampled to answer each  query, and $\lambda$ representing the total number of used data structures. The rest are the plain LSH construction with different values of $\lambda$.}
  \label{fig:nc-defense}
\end{figure*}

\begin{figure*}[h]
  \centering
  \includegraphics[width=\textwidth]{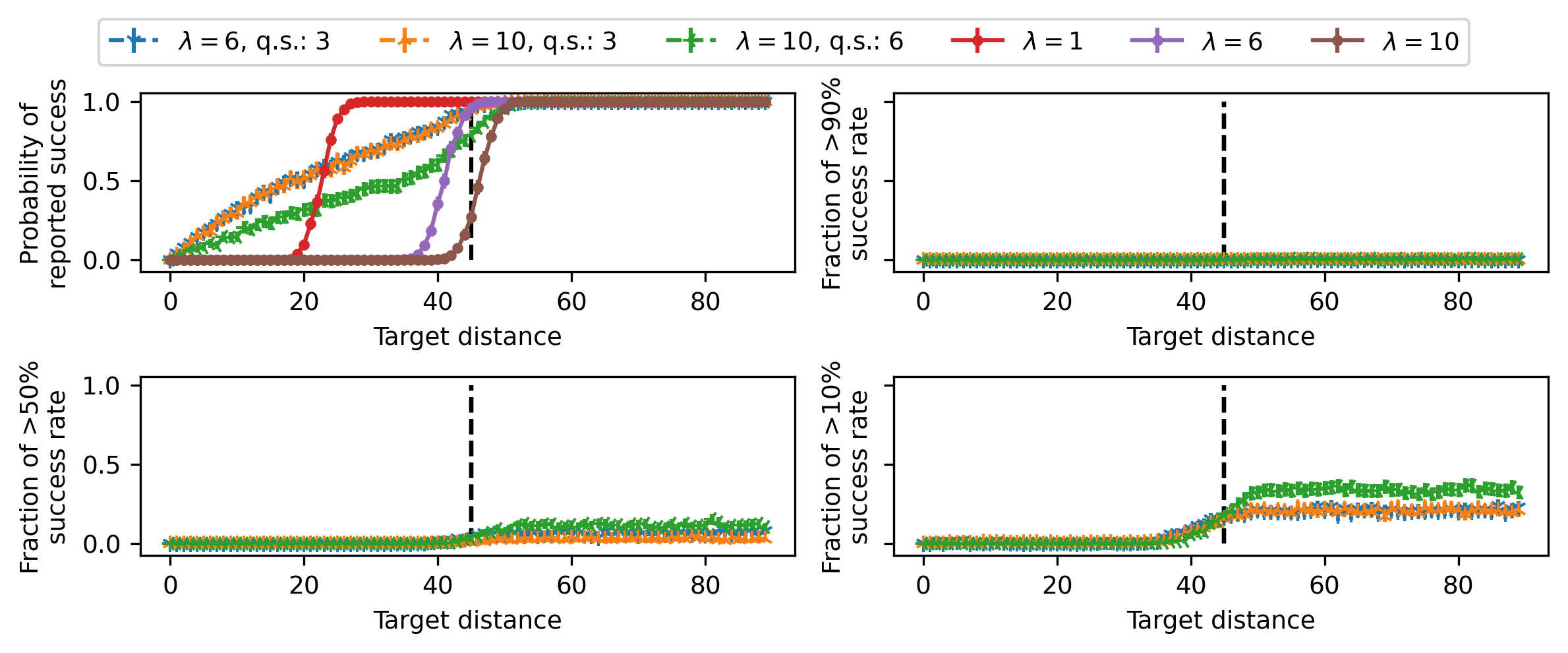}
  \caption{Performance of robustification based on differential privacy \cite{BeimelKMNSS22}. First subplot depicts probability that our adaptive adversary finds a false negative query versus the desired distance to that query from origin. The rest of subplots depict probabilities that, if such a query was found, if it was queried again it would still be a false negative query with probability at least $90\%$, $50\%$ or $10\%$. In the legend, plots annotated with both $\lambda$ and ``q.s'' correspond to the robust algorithm, with ``q.s.'' standing for ``query samples'', the number of data structures sampled to answer each  query, and $\lambda$ representing the total number of used data structures. The rest are the plain LSH construction with different values of $\lambda$.}
  \label{fig:dp-defense}
\end{figure*}

\clearpage

\end{document}